\documentclass[a4paper]{article}
\usepackage{amssymb}
\usepackage{graphicx,color}
\usepackage{amsthm}
\usepackage{amsmath}
\usepackage{enumerate}
\usepackage{tikz}

\usepackage{float}
\usepackage{subfig}
\usepackage{hyperref}
\usepackage{multirow}
\usepackage{rotating}
\usepackage{bbm}

\makeatletter
\newtheorem*{rep@theorem}{\rep@title}
\newcommand{\newreptheorem}[2]{%
\newenvironment{rep#1}[1]{%
 \def\rep@title{#2 \ref{##1}}%
 \begin{rep@theorem}}%
 {\end{rep@theorem}}}
\makeatother

\theoremstyle{plain}
\newtheorem{theorem}{Theorem}
\newreptheorem{theorem}{Theorem}
\newtheorem{proposition}[theorem]{Proposition}
\newtheorem{corollary}[theorem]{Corollary}
\newtheorem{lemma}[theorem]{Lemma}

\theoremstyle{definition}

\newtheorem{definition}[theorem]{Definition}

\def\ra{\rightarrow}

\def\ba{\begin{array}}
\def\ea{\end{array}}
\def\bi{\begin{itemize}}
\def\ei{\end{itemize}}
\def\mR{\mathbb{R}}
\def\mZ{\mathbb{Z}}
\def\mN{\mathbb{N}}
\def\mE{\mathbb{E}}
\def\m1{1}

\def\cV{\mathcal{V}}
\def\cE{\mathcal{E}}
\def\cZ{\mathcal{Z}}

\providecommand{\quotes[1]}{``#1"}

\let\oldhat\hat
\renewcommand{\vec}[1]{\mathbf{#1}}
\renewcommand{\hat}[1]{\oldhat{\mathbf{#1}}}

\providecommand{\com[2]}{\begin{tt}[#1: #2]\end{tt}}

\begin{document}
\title{Push sum with transmission failures}

\date{\today}

\author{Bal\'azs Gerencs\'er, Julien Hendrickx}
\author{Bal\'azs Gerencs\'er\thanks{B. Gerencs\'er and J. M. Hendrickx are with ICTEAM Institute,
  Universit\'e catholique de Louvain, Belgium
  {\tt\small balazs.gerencser@uclouvain.be} and {\tt\small
    julien.hendrickx@uclouvain.be} 
Their work is supported by the DYSCO Network (Dynamical Systems,
Control, and Optimization), funded by the Interuniversity
Attraction Poles Programme, initiated by the  Belgian
Federal Science Policy Office, and by the Concerted Research Action (ARC) of the
French Community of Belgium.}%
\and Julien M. Hendrickx\footnotemark[1]
}

\maketitle

\begin{abstract}
  The \emph{push-sum} algorithm allows distributed computing of the
  average on a directed graph, and is particularly relevant when one is restricted to one-way and/or asynchronous communications. We investigate its behavior in the presence of unreliable communication channels where messages can be lost.
  We show that exponential convergence still holds and deduce fundamental
  properties that implicitly describe the distribution of the final
  value obtained. We analyze the error of
  the final common value we get for the essential case of two nodes,
  both theoretically and numerically. We provide performance comparison with a standard
  consensus algorithm.
 \end{abstract}

\section{Introduction}
\label{sec:intro}

The ongoing active research on decentralized systems and distributed
computation constantly faces the mathematical challenge of aggregating
information spread across a huge network
\cite{jadbabaie2003coordination}, \cite{moreau2005stability},
\cite{cybenko1989dynamic}.
A fundamental case is
the task of averaging certain measurements obtained at the nodes of the
network which is known as the average consensus problem
\cite{tsitsiklis:phd1984}. The difficulty of this task heavily depends
on the assumptions on the communication network and the information available to the nodes.

In the simplest case, nodes communicate according to a
graph synchronously and without error. One method consists of the nodes exposing their
current values to their neighbors
and in turn updating their own values using a linear combination of
the values they have access to. The system is said to reach average consensus
if the values of all nodes converge to the average of the initial
measurements. Under suitable connectivity conditions, this can easily be achieved if the - possibly
time-varying - graph is undirected, as symmetric choice of weights will then keep the average constant. 
We will refer to these as
\emph{standard consensus} methods. They can be implemented under
particularly light requirements: nodes do not need unique identifiers
or global information about the system (e.g., to know the number of
nodes), and only use a single variable per node. Variations of
  this approach can be implemented on networks with directed
  asynchronous communications, but it is then in general impossible to
  keep the average constant over time, consequently the system does not
  converge to the initial average, see
  e.g., \cite{aysal2009broadcast,dimakis2010gossip,frasca2013large,frasca2013mean,
fagnani2008asymmetric} for analysis of the error for different such
algorithms. (Note that average preservation remains possible in
networks with directed communications if communications are
synchronous and pre-determined averaging weights satisfying a balance
condition are available).

The \emph{push-sum algorithm} (also known as ratio-consensus \cite{hadjicostis2014average}) allows computing the average on networks with directed asynchronous communications at the cost of one additional variable per node:
nodes not only record a linear combination of other nodes' values,
but also keep track of their ``relative importance'' in the system.
Therefore, two variables per node are required. The
algorithm can be asynchronous, with a randomly chosen node
communicating towards another at every step, without waiting for reply. This method
is known to efficiently compute the perfect average
\cite{kempe2003gossip,hadjicostis2014average}
and does not require agents to have unique identifiers or global information about the system.

We also note the existence of an interesting alternative approach based on so-called surplus
\cite{franceschelli2011distributed, wu2013broadcast, cai2012average,
  cai2011quantized}: nodes run a standard consensus
algorithm that does not necessarily preserve the average, and they
compensate their updates by removing the variation of their state
$x_i$ from an additional \quotes{surplus} variable $s_i$, so that the
sum of all states $x_i$ and surplus variables $s_i$  remains constant. An appropriate
mechanism then mixes the surplus variables, and slowly re-introduces
their values in the primary state variables $x_i$. This approach has been shown to
work in different context, but several open questions remain about the
general conditions under which it converges. 
Besides, 
all currently available implementations of this approach of which we are aware require using implicitly or explicitly 
global information about the system. 
To the best of our knowledge, every other available class of distributed algorithms allowing computing an average on networks with directed  asynchronous communications requires using unique node identifiers and/or global information about the network.

There is a strong research effort to understand the performance of these averaging methods
taking into account the deficiencies of the communication channels.
Such limitations include packet delays \cite{bliman2008average},
packet drops \cite{fagnani2009average},
changing connection topology \cite{ren2005consensus},
\cite{matei2013markovgraph}, limited or noisy communications
\cite{huang2009noisymeasurements}, \cite{bauso2009boundeddisturbance} 
or several of these \cite{nedic2010constrained}.
The issue we currently focus on is the presence
of transmission failures, resulting in some messages being lost.

In (asynchronous versions of) standard consensus protocols,
the emitter of a message does not change its state, so the loss of a
message is equivalent to the absence of communication. In
particular, when all messages have the same probability of being
lost, message losses only slow down the process without affecting
the distribution of the end result.
By contrast, the emitter of a message in a push-sum algorithm does
change its state, so the undetected loss of a message is not
equivalent to an absence of communication and will result in an
additional error.
A clever correction mechanism relying on sending the sum of all messages that would have been sent in the push-sum algorithm was proposed in \cite{vaidya2012robust}.
However, this method requires additional capabilities for the nodes; they need to have local
identifiers for each potential neighbor, and also an extra
memory slot for each communication link.

In this work, we restrict our attention to methods that do not use identifiers, large number of variables, or global information about the graph: We analyze the convergence and performance of the original push-sum algorithm in the presence of transmission failures without any corrective mechanism, similarly to what was done for asynchronous directed versions of standard consensus in
\cite{frasca2013large}. We also compare the error performance of the
push-sum algorithm and a comparable version of standard
consensus on directed networks with packet losses. Note that for push-sum,
errors are caused solely by packet losses, while for the other algorithm, errors are caused solely
by the directed character of the network.

Our contribution can be summarized as follows. First, we establish the
convergence of the push-sum algorithm in the presence of message
losses in Theorem~\ref{thm:convergence}. This ensures that the
algorithm does provide an
agreement among the nodes but does not yet give information on the
final value. The final value is possibly random as the algorithm
itself involves random elements, therefore a probability distribution
of the final value is expected instead of a deterministic number. The
resulting distribution is characterized in Theorem~\ref{thm:invariance}.

Bounds on the expected quadratic error of the final value compared to the real
average are developed for the fundamental case of two nodes. Theorem~\ref{thm:lowerbound}
provides lower bounds while Theorem~\ref{thm:generalupperbound} gives an upper bound.
New essential tools are introduced on the way in order to
deal with the non-standard features of the random process generated by
the push-sum algorithm.

Finally, the relevance of the push-sum algorithm is justified in the
presence of transmission failures: The performance is numerically compared with
that of a standard consensus algorithm, the asymmetric randomized gossip algorithm (ARGA) \cite{fagnani2008asymmetric}, and we indeed see that
push-sum is more efficient when the transmission failure rate is not
extremely high.

The rest of the paper is organized as follows. In
Section~\ref{sec:pushsum_intro} we formally describe the push-sum
algorithm. Results are stated in Section~\ref{sec:results}.
Section~\ref{sec:invariance} provides the tools for general
  understanding of the process which is needed to
perform our analysis. We then prove the error bounds in
Section~\ref{sec:boundproof} for two nodes.
Conclusions and further research directions are
discussed in Section~\ref{sec:conclusions}.

\section{The push-sum algorithm}
\label{sec:pushsum_intro}

In the general setting the goal of the push-sum algorithm is the following.
Given are $n$ agents with a strongly connected directed connection
graph and some input values at each agent, which we call \emph{initial measurements}. The agents want to
reach average consensus, all agreeing on the average of the initial measurements,
using only limited communication along the
edges of the graph.

First, we introduce the original version with perfect communication
from \cite{kempe2003gossip} then we present the setting considered in the current paper.

\begin{definition}
\label{def:pushsum_noerror}
The (gossip) \emph{push-sum algorithm with perfect communication} is the
following process. Given are a strongly connected directed graph $G$
and a distribution $\pi$ on the directed edges.
Every agent stores a value $x_i(t)$ and also an abstract
weight variable $w_i(t)$. The values $x_i(0)$ are the initial
measurements, which are to be averaged by the algorithm. The weights are initialized as $w_i(0)=1$ for
all $i$. At every time step, an edge $e=(i,j)$ is chosen from the
graph randomly and independently
according to $\pi$. As we work with asynchronous communication, we assume
only one edge is active at a time.
Agent $i$ sends half of both the value and
the weight to agent $j$. To be more precise,
the following update is performed:
\begin{equation}
\label{eq:pushsum_noerror}
\begin{aligned}
  x_i(t+1) &= x_i(t)/2, &  x_j(t+1) &= x_j(t)+x_i(t)/2,\\
  w_i(t+1) &= w_i(t)/2, &  w_j(t+1) &= w_j(t)+w_i(t)/2.
\end{aligned}
\end{equation}

\end{definition}

The behavior of the algorithm strongly depends on the connection graph and on the way
we pick the edges at every step. It is known to almost surely (a.s.) work well in a quite general setting
even with multiple simultaneous communications \cite{benezit2010weighted},
meaning that the rescaled values approach the average of the initial measurements at
every node:
$$\lim_{t\ra\infty}\frac{x_i(t)}{w_i(t)} = \frac{\sum x_j(0)}{n},
\qquad \forall i,~a.s.$$

In this paper we analyze the effect of transmission
errors: some messages sent might not reach their
destination, without the sender being aware of the failure. The
emitter would still update as if the transmission had been successful,
while the intended receiver would not, resulting in a loss of
information. 

\begin{definition}
\label{def:pushsum}
The \emph{push-sum algorithm with transmission failures}. The same
initial setup is used as in
Definition~\ref{def:pushsum_noerror}: given are a graph $G$, a distribution
$\pi$ on the edges, the initialization of the variables
$x_i(0),w_i(0)$. Additionally, for every edge $e$ a transmission
failure probability $p_e$ is given.
Edges $e=(i,j)$ are chosen the
same way according to $\pi$ at every step. However, this time a randomized update is performed: 
\begin{equation}
\label{eq:pushsum_werror}
\begin{aligned}
  x_i(t+1) &= x_i(t)/2, &  x_j(t+1) &= x_j(t)+\chi_e(t+1) x_i(t)/2,\\
  w_i(t+1) &= w_i(t)/2, &  w_j(t+1) &= w_j(t)+\chi_e(t+1) w_i(t)/2,
\end{aligned}
\end{equation}
where $\chi_e(s)$ are a collection of independent random variables taking a value $1$ with probability $1-p_e$ (successful transmission) and $0$ else (failed transmission).
\end{definition}

In the sequel we will study the asymptotic behavior of this system.
We will show that the rescaled values $x_i(t)/w_i(t)$ converge even though weights and
values go to 0, and then analyze the error between the limit
and the true average.
We propose general tools for
treating this process with transmission errors.
We exploit them in the simple case of two agents to get explicit error
bounds. 

\section{Results}
\label{sec:results}

\subsection{Convergence and error bounds}

Our first result, proved in Section~\ref{subsec:convergence}, ensures that the push-sum algorithm with transmission
failures still converges almost surely to some point in the convex
hull of the initial measurements, whenever $p_e<1$ (that is, there is
communication actually happening on the graph).

\begin{theorem}
  \label{thm:convergence}
The push-sum algorithm converges exponentially fast almost surely:
There exists some $q > 1$ and a random variable $x^*$ taking values in the convex
hull of $\{x_1(0),x_2(0),\dots, x_n(0)\}$ such that
$$
\limsup_{t \to \infty} \left|\frac{x_i(t)}{w_i(t)} - x^*\right|q^t \le
1,\quad a.s., \forall~i
$$
\end{theorem}

Therefore the algorithm does provide an
agreement among the nodes, but their final value is a random variable $x^*$ in general different from the average of the initial measurements.
This common final
value can be expressed as a convex combination of the initial measurements. This will be a consequence of Proposition \ref{prp:PSfn_coeff} in Section~\ref{sec:finalvalues}.

\begin{theorem}
\label{thm:finallinear}
The final value $x^*$ of Theorem~\ref{thm:convergence} is a random variable that can be re-expressed as 
$$
x^* = \sum_{i=1}^n \tau_i x_i(0),
$$
where the vector $(\tau_1,\tau_2,\ldots,\tau_n)$ is a random
variable that takes its values on the $n-1$-dimensional simplex: $\tau_i\geq 0$, $\sum_{i=1}^n \tau_i =1$,
and whose distribution is
independent of $x(0)$.
\end{theorem}

If $\tau_i = \frac{1}{n}$ for all nodes, then $x^*$ is the
exact average of the $x_i(0)$. We take as measure of performance of
the algorithm the expected square error on $\tau$ with respect to this ideal case: 
\begin{equation}
  \label{eq:Rdef}
  R = \mE\left(n \sum_{i=1}^n \left(\tau_i - \frac{1}{n}\right)^2\right).  
\end{equation}
In particular, $R=0$ exactly when $\tau =
(1/n,1/n,\ldots,1/n)$ and $R$ is maximal in the worst case of $\tau = (1,0,\ldots,0)$ corresponding to the final value being influenced by only one single node, when it becomes $n-1$. If $\tau$ was uniformly chosen
from the simplex, $R$ would be $1-\frac{4n}{(n+1)(n+2)}$. One can also verify that $R/n$ corresponds to the expected square error if the $x_i(0)$ are scalar i.i.d. random variables drawn from a distribution with variance $1$. (We chose $R$ and not $R/n$ as error measure because one can verify that $R$ remains invariant when several identical $\tau$ are joined with identical scaling.)

\begin{figure}[h]
  \centering
  \includegraphics[width=0.6\textwidth]{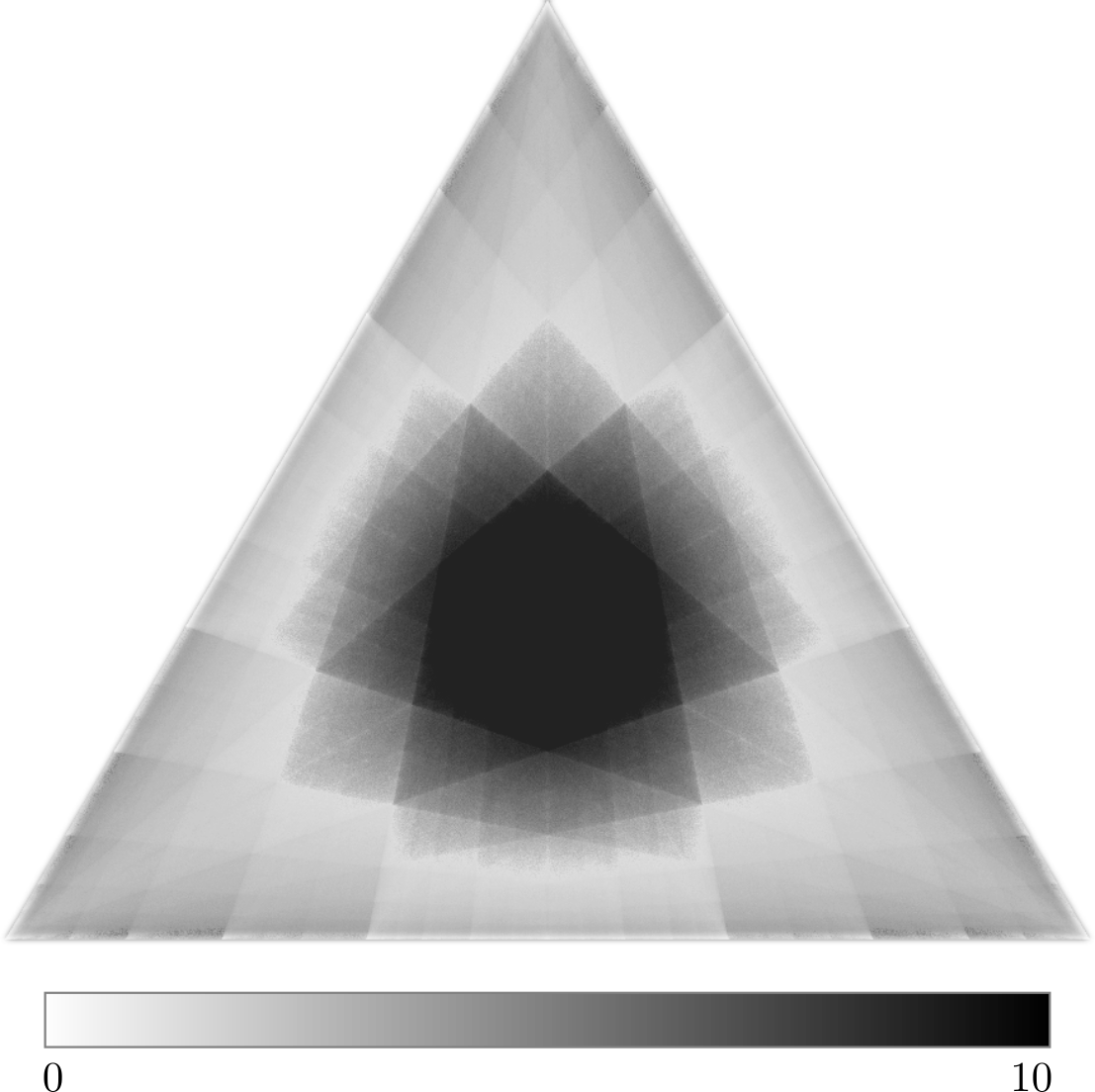}
  \caption{Simulation of the density of the distribution $\mu$ of the
    vector $\tau$
     of Theorem \ref{thm:finallinear}, for $n=3$
    nodes fully connected and transmission failure probability
    $p_e=0.6$ on all edges.}
  \label{fig:hist06}
\end{figure}

In order to derive bounds on $R$ a key tool of our analysis is to understand
the distribution $\mu$ of the
vector $\tau$. This distribution is relatively complex, as suggested
by Figure~\ref{fig:hist06}, where we present a simulation of the density of
$\mu$ for fully connected 3 nodes with
$p_e=0.6$ on all edges (using 500000 samples with the density displayed on a
logarithmic scale). Therefore it is supported on the triangle spanned by $(0,0,1),(0,1,0),(1,0,0)$.
We will see that the distribution $\mu$ can be characterized using an
invariance relation:
\begin{equation}
  \mu = \mE F\mu,
  \label{eq:muinvariance}
\end{equation}
where $F$ is a random measure transformation based on a single
step of the push-sum algorithm.
The random measure transformation $F$ is fully defined in
Section~\ref{subsec:measures} using the notations introduced on
the way and Theorem~\ref{thm:invariance} provides the precise statement
formalizing \eqref{eq:muinvariance}.

We proceed to derive bounds on this final
value. This is done for the fundamental case of 2 nodes, where the
edges in the two directions are selected with equal probability and
the transmission failure probability is the same for them, denoted
simply by $p$ from now. The proofs
for these bounds are in Section~\ref{sec:boundproof}.

\begin{theorem}
  \label{thm:lowerbound}
   Let $\phi=\frac{1-\sqrt{1-p^2}}{p}$. 
  The following lower bounds
  hold for the expected quadratic error $R$ for 2 nodes having equal probabilities of initiating a transmission:
  \begin{equation}
    \label{eq:lowerbound}
    R\ge \phi - 4(1-\phi)\sum_{i=1}^\infty
    \frac{(2\phi)^i}{(2^i+1)^2}\ge \phi -
    \frac{8}{9}\phi(1-\phi)-\frac{2}{2-\phi}\phi^2(1-\phi).    
  \end{equation}
\end{theorem}
For $p\ra 0$, both lower bounds are asymptotically
$\frac{p}{18} + O(p^2)$.

\begin{theorem}
  \label{thm:generalupperbound}
  The following upper bound
  holds for the expected quadratic error $R$ for 2 nodes having equal probabilities of initiating a transmission:
  \begin{equation}
    \label{eq:generalupperbound}
    R \le \frac{p(1-p)^2}{3+p} + \frac{p}{25(1+p^2)}\left(18+23p+50p^2-41p^3\right).
  \end{equation}
\end{theorem}

The ideas that we use could in principle be applied to systems with three or
more nodes, but the derivations become much more involved, partly due
to the more complex distribution of $\tau$ generated by our algorithm
on higher dimensional simplexes.

\subsection{Simulations}

Three types of simulations have been performed. First, to evaluate the
numerical performance of the theoretical bounds obtained. Second, to
get an overview of the behavior of the algorithm for larger graphs
which are currently out of reach with our analytic results. Finally,
to compare the push-sum concept with ARGA (asymmetric randomized gossip algorithm \cite{fagnani2008asymmetric}) as a tool
for distributed averaging.

In view of Theorem~\ref{thm:finallinear}, we can understand the
precision of the algorithm via the expected quadratic error $R$ of the random variable $\tau$ defined in \eqref{eq:Rdef}.
This corresponds to launching the algorithm with the basis vectors
$e_i=(0,\ldots,0,1,0,\ldots)$ as initial values at the nodes of the
graph.

For the computations, the algorithm is executed until approximate
convergence is achieved which is formulated using the following stopping criterion: We
require that for each coordinate of the coefficient
vector, which would eventually converge to $\tau$, the nodes should
all agree within a factor of $1.0001$. There is also a step count
limit when the simulation is discarded if agreement is not reached
within this bound. This is set at
$1M$ when the nodes are fully
connected and at $5M$ for the
supposedly slower grid and cycle topologies.
These bounds are set to avoid the rare but extremely long instances
and were chosen high enough to be barely reached (and possibly cause a
bias on the result) except when $p_e>0.95$
on all edges.

For $p_e<1$ message losses drive all value and weight variables $x_i(t)$
and $w_i(t)$ to 0, which might cause quantization problems
on a digital computer, but floating point representation of numbers
avoids such issues. Nevertheless, when implementing the algorithm for much larger graphs,
additional attention has to paid to avoid such issues.

The numerical performance of the lower and upper bounds proven in
Theorem~\ref{thm:lowerbound} and \ref{thm:generalupperbound} is shown in Figure~\ref{fig:lowerupperbounds}.
\begin{figure}[h]
  \centering
  \resizebox{0.6\textwidth}{!}{\input{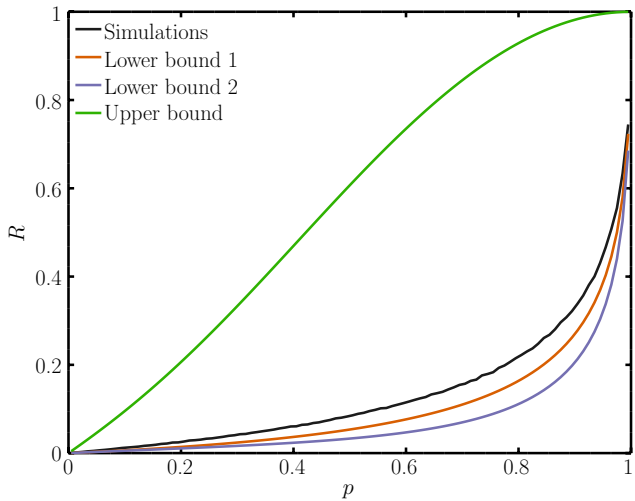}}
  \caption{Performance of the lower and upper bounds on the expected quadratic
    error $R$ for two nodes stated in Theorem~\ref{thm:lowerbound} and Theorem~\ref{thm:generalupperbound} for different transmission failure
    probabilities $p$, compared with simulation data based on 3M
    samples. Lower bound 1 and 2 correspond to the middle and final
    expressions in the inequality of Theorem~\ref{thm:lowerbound}, respectively.}
  \label{fig:lowerupperbounds}
\end{figure}
We can see that both lower bounds follow well the numerical
values obtained through the whole range of $p$. However, the upper
bound is visibly rather conservative. It does capture the linear
nature of the error when $p\ra 0$ but the numerical values are still
far.

We also investigate the dependence of the error on the transmission
failure probability for larger
graphs. We assume the graph is homogeneous with this respect, having
the same failure probability $p$ on all edges. We consider the cycle, the grid, and the complete graph topologies for
different node counts (to ensure homogeneity, the grid is arranged as a torus, with left and
right, top and bottom edges connected, respectively). The data is
shown in Figure~\ref{fig:err_clique_summary}, \ref{fig:err_grid_summary} and \ref{fig:err_cycle_summary}.
\begin{figure}[h]
  \centering
  \resizebox{0.6\textwidth}{!}{\input{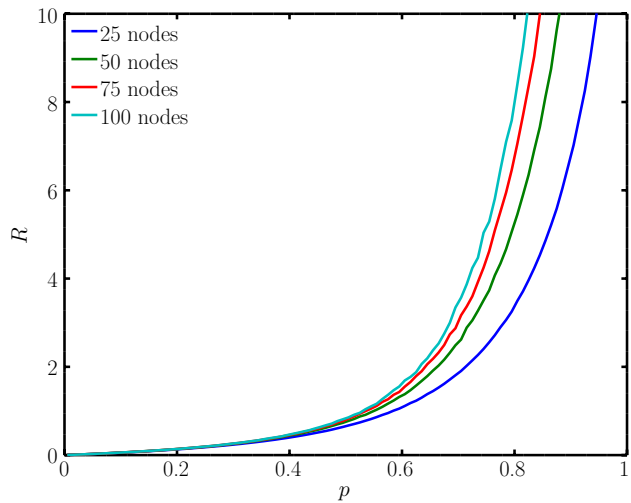}}
  \caption{Evolution of the expected quadratic error $R$ with the transmission failure
    probability $p$ for complete graphs of different sizes, based on
    simulation data of 10M samples.}
  \label{fig:err_clique_summary}
\end{figure}
\begin{figure}[h]
  \centering
  \resizebox{0.6\textwidth}{!}{\input{error_grid}}
  \caption{Evolution of the expected quadratic error $R$ with the transmission failure
    probability $p$ for grid graphs of different sizes, based on
    simulation data of 3M samples.}
  \label{fig:err_grid_summary}
\end{figure}
\begin{figure}[h]
  \centering
  \resizebox{0.6\textwidth}{!}{\input{error_cycle}}
  \caption{Evolution of the expected quadratic error $R$ with the transmission failure
    probability $p$ for cycle graphs of different sizes, based on
    simulation data of 1.5M samples.}
  \label{fig:err_cycle_summary}
\end{figure}

We see that for the grid and the complete graph the two plots are
qualitatively similar. However, for the cycle graph we get a
distinctly different picture.
Considering the cycle graph, a scaled
version of the error $R/n$ is plotted in
Figure~\ref{fig:err_cycle_scaled} in order to reveal our
numerical finding, that for fixed $p$, $R/n$ appears to remain
constant with $n$ for large enough cycles. The scaled error
$R/n$ for the other networks did not exhibit any specific
behavior.

\begin{figure}[h]
  \centering
  \resizebox{0.6\textwidth}{!}{\input{error_cycle_scaled}}
  \caption{Evolution of $R/n$, the expected quadratic error after
    scaling with the transmission failure
    probability $p$ for cycle graphs of different sizes, based on
    simulation data of 1.5M samples.}
  \label{fig:err_cycle_scaled}
\end{figure}

Following the analysis of the push-sum algorithm, we compare its efficiency to
the ARGA \cite{fagnani2008asymmetric}, which is a variation of
  the standard consensus methods, and can be implemented in exactly
  the same condition as the version of push-sum we study: It requires
  no identifiers,  global knowledge, or large number of variables. At
  every time step, an $(i,j)$ edge is chosen randomly, then node $i$
  sends a message to $j$. But nodes store only one variable
  $x_i(t)$, and the update after $i$ has sent its value to $j$ is
\begin{align}
\label{eq:consensus_def}
  x_i(t+1) &= x_i(t), & x_j(t+1) &= (x_j(t)+x_i(t))/2.
\end{align}
It is easy to see that this also converges to consensus. However, even
with no communication error the consensus value might deviate from the
real average \cite{fagnani2008asymmetric}. On the other hand, transmission failures will not increase
this error if the probability is the same for all messages; the will only slow down the process.

For a fair comparison it will be useful to allow a slight adjustment of
both algorithms. In both cases, we have a hard-coded parameter of
influence set at $1/2$. This is the ratio sent for the push-sum
algorithm and the strength of influence for ARGA. We
get valid and convergent algorithms if we change this value to some
other $0<\alpha<1$.

This way the algorithm has two parameters: first, the external
transmission failure
probability $p$ which we again assume to be the same on all edges and
second, the chosen influence ratio $\alpha$. We also have
two performance metrics: the error
from the real average and the speed to reach consensus. Therefore we
get a legitimate comparison in the following way. Given an error
probability $p$ and a desired error $R$, we choose $\alpha$ for both
algorithms separately to achieve this error. Knowing that now they reach the same level of accuracy, comparing the speed of the two will tell us
which one is more efficient.

Therefore we generate a large number of instances, randomly choosing
$p$ and $\alpha$ and compute the speed and error the algorithms
give. We generated over 60M instances for the push-sum algorithm, and
over 35M for ARGA (for the latter we can slightly simplify
as $p$ simply scales the speed without changing the error,
we can fix it to 1 and then scale it to the values needed).
In Figure~\ref{fig:speedcomparison} we present the
performance comparison results of the two algorithms for a network of fully connected 5
nodes. The parameter space of $(p,R)$ is split into 3 regions:
\begin{itemize}
\item[a)] ARGA performs faster. This is the case when
  we want low error despite the extremely high $p$.
\item[b)] Push-sum is more efficient.
\item[c)] None of the random instances of the push-sum algorithm ended
  up in this region.
\end{itemize}

At the boundary between region $a)$ and $b)$ black identifies the
points where the two algorithms perform equally, the surrounding grayscale corresponds
the ratio of the speeds, becoming white when it is above $1.1$.

To understand region $c)$ recall that the push-sum algorithm will
always give 0 error when $p=0$, so it is reasonable to expect that
no high errors occur when $p$ is very small. There is also a
small area near $p=1,R=0$ which means it is very unlikely to get near the
true average when the transmission error probability is extremely
high.

It is important to add that on the area on the upper left
corner, push-sum is a preferable choice compared to ARGA in the
following way. When some
choice of parameters is unreachable for push-sum (e.g.,
$p=0.2, R=0.4$), we may use a feasible instance with stronger error
constraints ($p=0.2, R=0.25$ in this case) and it turns out that this
is still faster than ARGA for the original parameters. This
consistently holds for the whole part of region $c)$ corresponding to
low $p$ and high $R$, found to the left and above
region $b)$.

\begin{figure}[h]
  \centering
  \resizebox{0.6\textwidth}{!}{\input{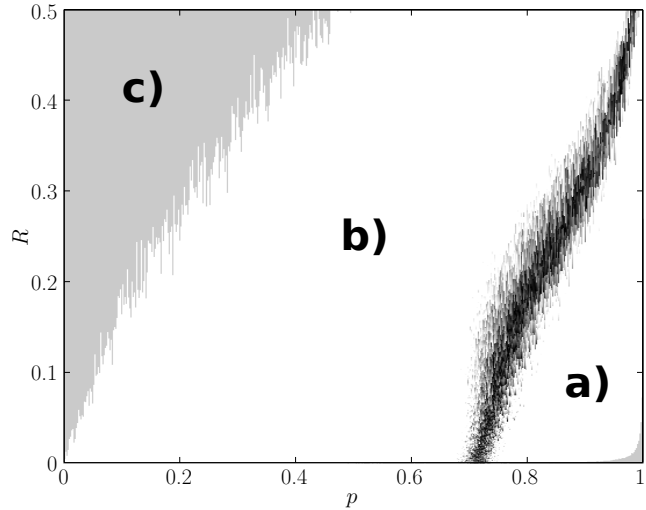}}
  \caption{Regions of the parameter space $(p,R)$ of transmission
    failure probability and target error tolerance where a) consensus
    is more efficient; b) push-sum is more efficient; c) push-sum
    does not reach this zone.}
  \label{fig:speedcomparison}
\end{figure}

We conclude that the push-sum algorithm is the better alternative for
a wide range of setups, but we should be careful when the
transmission failure probability $p$ is extremely high.

\section{The invariance relation}
\label{sec:invariance}

In this section we will prove an invariance relation that will be of
key importance for our analysis. Before that, we need to
establish fundamental properties of the push-sum algorithm with transmission
failures.

\subsection{Convergence}
\label{subsec:convergence}

First we prove Theorem~\ref{thm:convergence} showing that the
push-sum algorithm with transmission failures converges exponentially fast, meaning that the 
scaled values of different agents will be the same in the
limit. It is easy to verify that all values and weights go to zero for
$0<p_e<1$ because the same fraction of loss may happen independently at
every step. Nevertheless
their ratio converges to a meaningful value as shown below. At this point we only aim to 
confirm convergence, and do not search for the best rate possible.
\begin{proof}[Proof of Theorem~\ref{thm:convergence}]
  First we prove convergence for scalar $x_i(t)$. For higher dimensions, we can use this result coordinate-wise, then we will need to additionally verify that $x^*$ indeed lies in the convex hull of the initial values.

  The proof relies on analyzing the following event: a node $i$ holding the highest weight $w_i(t)$ at time $t$ successfully influences, maybe indirectly, all other nodes over a period of fixed length starting at $t$. We will show that (i) this event contracts the range of all ratios $x_j/w_j$ by a certain fixed proportion, and (ii) it almost surely occurs sufficiently often. 

To formalize this idea, let $i^*_t$ be (one of) the nodes with the highest weight at time $t$. Since the communication graph is strongly connected, it is therefore possible to select a directed tree starting from $i^*_t$ and reaching every other node. Let $e^t_1,e^t_2,\ldots, e^t_{n-1}$ be the edges of this tree, ordered starting from node $i_t^*$ and moving away (e.g.,
following a depth-first search). The broadcast event $B(t)$ consists of the edges $e^t_1,e^t_2,\ldots, e^t_{n-1}$ being selected at respective times $t,t+1,\ldots,t+n-2$ (and thus no other edge being selected during that period) and all transmissions being successful. The core of the proof relies on the following two claims:

{\bf Claim 1.} Let \begin{align*}
    M(t) &:= \max_k \frac{x_k(t)}{w_k(t)}, & m(t) &:= \min_k \frac{x_k(t)}{w_k(t)}.
  \end{align*}
  be the maximal and minimal ratios at time $t$. If the event $B(t)$ occurs, then 
    $$M(t+n-1)-m(t+n-1) \le q' (M(t)-m(t)),$$
    for the constant $q'=1-\frac{1}{2^n}$.
\begin{proof}
Observe the algorithm is invariant under the addition of $Kw_i$ to
each $x_i$ for a common $K$, meaning that it corresponds to a translation of all
ratios $\frac{x_i}{w_i}$ by $K$, and this translation is
preserved by the algorithm. Hence we can assume without the loss of generality that $m(t) = 0 = \min_i x_i(t)$. Let $w^*=w_{i_t^*}(t) = \max_i w_i(t)$, $x^* = x_{i_t^*}(t)$, and $j_s$ be the end node of $e^t_s$. We now show 
that after node $j_s$ is reached at step $s$, for all $u \geq s \geq t$, there holds
\begin{align}
x_{j_s}(u) &\geq \frac{1}{2^{u-t}}x^*,   \label{eq:ineq_x}\\
w_{j_s}(u) &\leq 2w^*.   \label{eq:ineq_w}
\end{align}
We show this by induction, first on $u$ keeping $s$ fixed, then increasing $s$.
To start, for $s=t$ these inequalities hold by definition for $u=t$. Afterwards, if they hold for a certain $u, s$, they further hold for the same $s$ and all $u' \geq t$ because at every time step, $x_{j_s}$ and $w_{j_s}$ either are divided by 2 or remain unchanged.

Given $s>t$, let us now assume the inequalities hold for all $s>s'\geq
t$ and all $u\geq s'$, and prove they then hold for $s$ and $u=s$. Node $j_s$ has not emitted or received any message since time $t$ prior to $s$, hence $x_{j_s}(s-1) = x_{j_s}(t) \geq 0$ and $w_{j_s}(s-1) \leq w_{j_s}(t) \leq w^*$. At time $s$, it receives a message from a node $j_{s'}$ from some $s'<s$. Using \eqref{eq:ineq_w}, there holds 
$$
w_{j_s}(s) = w_{j_s}(s-1) +\frac{1}{2} w_{{j_{s'}}}(s-1) \leq w^* + \frac{1}{2}2w^*= 2w^*.
$$ 
Similarly, using \eqref{eq:ineq_x}, there holds
$$
x_{j_s}(s) = x_{j_s}(s-1) + \frac{1}{2} x_{{j_{s'}}}(s-1)\geq 0 + \frac{1}{2}\cdot \frac{1}{2^{s-1-t}} x^* = \frac{1}{2^{s-t}} x^*.
$$
Inequalities \eqref{eq:ineq_x} and \eqref{eq:ineq_w} are now proved for the increased $s$ and $u = s$, then $u$ can be incremented as before.
Now relying \eqref{eq:ineq_x} and \eqref{eq:ineq_w} for any $j$ we have
$$
\frac{x_j(t+n-1)}{w_j(t+n-1)} \geq \frac{\frac{1}{2^{n-1}x^*}}{2w^*}= \frac{1}{2^n}\frac{x_{i^*}(t)}{w_{i^*}(t)}.
$$

Let us re-introduce the possibility of a $m(t)\neq 0$, retranslating all the ratios $\frac{x_j}{w_j}$ simply yields
$$
\frac{x_j(t+n-1)}{w_j(t+n-1)}-m(t) \geq  \frac{1}{2^n}\left( \frac{x_{i^*}(t)}{w_{i^*}(t)}-m(t)  \right),
$$
for all $j$. Minimizing over $j$ gives
$$
m(t+n-1)-m(t) \geq  \frac{1}{2^n}\left( \frac{x_{i^*}(t)}{w_{i^*}(t)}-m(t)  \right).
$$
A similar argument shows
$$
M(t)-M(t+n-1) \geq  \frac{1}{2^n}\left(M(t)- \frac{x_{i^*}(t)}{w_{i^*}(t)}\right).
$$
Adding up the two, after cancellations we arrive at
$$
M(t+n-1) - m(t+n-1)\leq \left(1-\frac{1}{2^n}\right)(M(t)-m(t)).  
$$
\end{proof}
 
{\bf Claim 2.}
  \label{cl:N_growth}
  Counting the occurrences of $B(t)$, let $N_k = \big|\{1\le l \le k ~:~ B(l(n-1))~ \text{occurs} \}\big|$. Then $N_k$ grows asymptotically at a linear rate, namely
  $$\liminf_k \frac{N_k}{k} \ge \alpha,\qquad a.s.,$$
  for some $\alpha>0$.

\begin{proof}
Note that the events $B(l(n-1))$ are not independent, as the occurrence of
one will possibly change which node $i^*$ will have the highest weight
at a further time and hence the probability of the occurrence of $B(t)$ at that time.
Nevertheless, their probability conditional to past history can be uniformly bounded. 
Let $\underline \pi= \min \{\pi_{e}>0\}$ be the probability of selection of the least probable edge, and 
$\overline p= \max \{p_{e}<1\}$ 
the maximal transmission failure probability. The event $B(t)$
consists of a certain $n-1$ long sequence of edges to be selected in a
given order starting at time $t$, and of all the corresponding
transmission to be successful. Its probability conditional to past
history is thus at least $\alpha = (\underline \pi (1-\overline p))^{n-1}$.

We need one more step to avoid dependency issues. By time $t$ the probability of $B(t)$ is known (as $i^*(t)$ is identified). Let us introduce an independent confirmation event $B_c(t)$ of probability $\alpha/P(B(t))$ and define $\tilde B(t)$ as the event that both $B(t)$ and $B_c(t)$ occur. Clearly $P(\tilde B(t)) = \alpha$ and $\tilde B(l(n-1))$ is an i.i.d.\ time series.
Using $\tilde N_k = \big|\{1\le l \le k ~:~ \tilde B(l(n-1))~ \text{occurs} \}\big|$, we have $N_k \geq \tilde N_k$, and standard concentration results show  $\liminf_k  \frac{ N_k}{k}\geq \liminf_k  \frac{\tilde N_k}{k} = \alpha$, a.s.
\end{proof}

Let us turn back to proving Theorem \ref{thm:convergence} for scalar $x_i(t)$.
  Previously, in Claim 1 we have shown the contracting effect of $B(t)$ on the
  range of the ratios, by a factor $q'<1$ in $n-1$ steps. In Claim 2 we have just shown that this happens
  often enough, approximately $\alpha k$ times out of $k$ possibilities. Combining the two claims we get
  \begin{equation}
    \label{eq:Mm_decrease}
    \limsup_k \Big(\big(M(k(n-1)) - m(k(n-1))\big)
    (q'^{-\alpha}-\varepsilon')^k \Big)\le 1, \qquad a.s.
  \end{equation}
  for any small $\varepsilon'>0$.
  The statement of the theorem follows by setting
  $$x^* = \bigcap_t[m(t),M(t)], \qquad q = q'^{-\frac{\alpha}{(n-1)}}-\varepsilon,$$
  for any small $\varepsilon>0$.

  Finally, for multidimensional $x_i(t)$, \eqref{eq:Mm_decrease} holds for each coordinate and this implies exponential convergence in norm with rate $q$.
  We have to ensure that the limit point $x^*$ lies in the convex hull of the initial values. Observe that the convex hull of $\{\frac{x_1(t)}{w_1(t)},\frac{x_2(t)}{w_2(t)},\dots, \frac{x_n(t)}{w_n(t)}\}$ can only shrink: the only situation when a ratio changes is when some node $j$ successfully receives a message from some other node $i$. In that case the new ratio is
  $$\frac{x_j(t+1)}{w_j(t+1)} = \frac{x_j(t)+x_i(t)/2}{w_j(t)+w_i(t)/2} = \frac{w_j(t)}{w_j(t)+w_i(t)/2} \cdot\frac{x_j(t)}{w_j(t)} + \frac{w_i(t)/2}{w_j(t)+w_i(t)/2} \cdot\frac{x_i(t)}{w_i(t)}.$$
  This new ratio is a convex combination of the previous ones, so the convex hull at $t+1$ has to be inside the convex hull at $t$. The intersection of this descending sequence can only be $\{x^*\}$ and the initial body at $t=0$ is the convex hull of $\{\frac{x_1(0)}{w_1(0)},\frac{x_2(0)}{w_2(0)},\dots, \frac{x_n(0)}{w_n(0)}\}$. Note that all $w_i(0)$ are initialized as 1, so we can omit dividing by them to complete the proof.
\end{proof}

\subsection{Final values}
\label{sec:finalvalues}

In the previous subsection we have seen that the push-sum algorithm
always reaches consensus. However, due to the random transmissions and
failures this final value is a random variable. We set up the definitions to study the behavior of the final consensus value acquired from the push-sum algorithm with
transmission failures. 
We introduce two equivalent descriptions of the algorithm to formulate this,
first we define a random function
$T$ as the \emph{push-sum function} which results in the final consensus value of a push-sum algorithm
given the initial measurements of the agents. We will then define the
\emph{push-sum coefficient} which captures how the push-sum algorithm
mixes the initial measurements.

For the remaining of this subsection, we fix a directed graph $G$ on the $n$ nodes together with a distribution $\pi$ on the edges, and a transmission failure probability $p_e$ for every edge. 

\begin{definition}
\label{def:PSoperatore}
For any real vector space $V$ we define the \emph{push-sum step} to
be a random operator $S:V^n\times\mR^n\rightarrow
V^n\times\mR^n$, with $(v',w')= S(v,w)$ being the result of the
selection of an edge $e=(i,j)$ according to the distribution $\pi$,
then for every $k\neq i,j$ set $v_k'=v_k,w_k'=w_k$ 
and also
\begin{align*}
  v_i' &= v_i/2, &  v_j' &= v_j+\chi_e v_i/2,\\
  w_i' &= w_i/2, &  w_j' &= w_j+\chi_e w_i/2.
\end{align*}
where $\chi_e$ is a random variable taking a value $1$ with probability $1-p_e$ (successful transmission) and $0$ else (failed transmission).
\end{definition}

\begin{definition}
\label{def:PSfunc}
For any real vector space $V$ we define the \emph{push-sum function}
as the random function $T:V^n \rightarrow V$ whose value $T(\cV)$ for
some $\cV =(v_1,v_2,\ldots, v_n)$ is obtained in the following way: 

We initialize
$\cV(0)=\cV,~w(0)=\vec{1}\in\mR^n$, and at every time step $t$ we apply an independent push-sum step 
$$
(\cV(t+1),w(t+1)) = S(\cV(t),w(t)).
$$
We compute then the value of the push-sum function as 
\begin{equation}
  \label{eq:PSfunc_final}
  T(\cV) = \lim_{t\rightarrow\infty}\frac{v_i(t)}{w_i(t)}.
\end{equation}
Observe that this is well defined almost surely as the limit is the same for all
indices $i$. This follows from the convergence of the push-sum
algorithm (Theorem~\ref{thm:convergence}). Still, the function value is random, depending on the
choice of indices and $\chi_e$ at each time step.

Defining the operator $N:V^n\times\mR^n\rightarrow V$ by $N(\cV,w) =
\frac{v_1(t)}{w_1(t)}$ (where the choice of the index 1 here is arbitrary), we can rewrite the value taken by push-sum function $T$ as
\begin{equation}
  \label{eq:PSfunc_compact}
  T(\cV) = \lim_{t\rightarrow\infty}N(S^t(\cV,\vec{1})),
\end{equation}
where with a slight abuse of notation $S^t$ means the composition of
$t$ independent copies of $S$.
\end{definition}

Observe that the push-sum function is linear in its input. We therefore decompose it as a random linear operator applied to the
input vectors. For this purpose, we introduce the push-sum coefficients.

\begin{definition}
\label{def:PScoeff}
Let $\cE = (e_1,e_2,\ldots,e_n)$
 with $e_i = (0,0,\ldots,1,\ldots,0)^\top$ having a
  single 1 at coordinate $i$. 
The \emph{push-sum coefficient} $\tau\in\mR_+^n$ is the random variable  $\tau = T(\cE)$.

We will also use the random variables $\cZ(1),w(1)$ corresponding
respectively to the values and weights after one single step of the
push-sum algorithm starting from $\cE$ and unit weights. Formally:
\begin{equation}\label{eq:defcZ}
(\cZ(1),w(1)) = S(\cE,\vec{1}); 
\end{equation}
\end{definition}

The following Proposition confirms the interpretation of the push-sum coefficients. It directly implies Theorem \ref{thm:finallinear}.

\begin{proposition}
\label{prp:PSfn_coeff}
Let  $\cV=(v_1,v_2,\ldots,v_n)\in V^n$ for a vector space $V$. There holds  
$$T(\cV) = \langle \tau,\cV\rangle ,$$
where $\tau$ is a push-sum coefficient of Definition \ref{def:PScoeff}.
\end{proposition}

  \begin{proof}
    We run two push-sum algorithms computing $T(\cV)$ and $T(\cE)$ in
    parallel, coupling the choices of edges and transmission failures
    of the two. We denote by $v_i(t)$ the vectors for
    the process computing $T(\cV)$ and we denote by $z_i(t)$ the vectors
    used for $T(\cE)$. The weights $w_i(t)$ will be the same for the
    two processes.

    The variables are initialized as $v_i(0) = v_i$ and $z_i(0) =
    e_i$. At start we clearly have $v_i(0) = \langle z_i(0), \cV \rangle$.
    It is easy to see that the system of equation $v_i(t) = \langle
    z_i(t), \cV \rangle$ is preserved at
    each step. Adding the fact that the weight variables are the same
    for the two processes we get
    $$
    \frac{v_1(t)}{w_1(t)} = \left\langle \frac{z_1(t)}{w_1(t)}, \cV \right\rangle.
    $$
    When taking the limit in $t$ the definitions of $T(\cV)$ and
    $T(\cE)=\tau$ appear which completes the proof.
  \end{proof}

\subsection{Measures and evolution of distributions}
\label{subsec:measures}

We plan to work with the distribution of the final value of the
push-sum algorithm, for which we introduce the following measures and transformation.

\begin{definition}
  Let $\mu$ be the distribution of a push-sum coefficient $\tau = T(\cE)$. This
  measure has a support contained in the unit simplex $\{x\in\mR_+^n~|~\|x\|_1=1\}$.
\end{definition}

\begin{definition}
We define the induced measure transformation $\langle \sigma,\cV\rangle $ for a measure $\sigma$ on
$\mR^n$ and $\cV\in V^n$. This is a measure on $V$, formally
$$\langle \sigma,\cV\rangle (A) = \sigma\big(\{x\in\mR^n~|~\langle x,\cV\rangle \in A\}\big).$$
In other words, for any set $A\subset V$ we identify and measure those
$x$ for which $\langle x,\cV\rangle$ falls in $A$.
\end{definition}

\begin{definition}
We define the rescaling operator $L$ on $\mR_+^n\setminus\{0\}$ as 
$$L(x) = \frac{x}{\|x\|_1}.$$
\end{definition}
This operator scales its argument to have unit $1$-norm. The operator
$L$ naturally induces the measure transformation $L^*$,
$$L^*(\sigma)~(A) = \sigma\big(\{x\in\mR_+^n\setminus\{0\}~|~L(x) \in A\}\big).$$

Now we develop an invariance relation describing
$\mu$ using the time-homogeneity of the push-sum algorithm. Broadly
speaking, we want to exploit that a push-sum algorithm is equivalent to a push-sum step followed by a push-sum algorithm.

The main result of this subsection is the following Theorem,
providing a formal description of
\eqref{eq:muinvariance} stated in Section~\ref{sec:results}.

\begin{theorem}
  \label{thm:invariance}
  After the first step of a push-sum coefficient algorithm we get the random
  $n$-tuple of vectors $\cZ(1)=(z_1(1),z_2(1),\ldots,z_n(1))$ defined in \eqref{eq:defcZ}. Then the
  following invariance equation holds for the push-sum coefficient distribution $\mu$:
  \begin{equation}
    \label{eq:PSinvariance}
    \mu = \mE L^*\left(\langle \mu,\cZ(1)\rangle\right).    
  \end{equation}
\end{theorem}
 To see the right hand side of \eqref{eq:PSinvariance} in more detail, observe that we first take the inner product of $\mu$
  with the (random) collection of vectors $\cZ(1)$. We map back the
  resulting measure to the simplex of unit 1-norm non-negative
  vectors. Finally we take expectation w.r.t.\ $\cZ(1)$.

\begin{proof}
  The random variable $\tau=T(\cE)$ follows the distribution $\mu$.
  We launch the
  push-sum algorithm as specified in
  Definition~\ref{def:PSfunc}. We may take $t+1$ push-sum steps by first taking one, then
  applying all the remaining steps to the result. Formally,
  \begin{equation}\label{eq:exprTcZSZ1}T(\cE) = \lim_{t\ra\infty} N(S^{t+1}(\cE,\vec{1})) = \lim_{t\ra\infty}
  N(S^t(\cZ(1),w(1))).\end{equation}
  We define a very similar random variable for comparison. We perform
  a single step of the push-sum algorithm, which gives us $\cZ(1)$. We
  now treat this as the input for a new push-sum function. Formally,
  we are looking at
\begin{equation*}
T(\cZ(1)) = \lim_{t\ra\infty} N(S^t(\cZ(1),\vec{1})).
\end{equation*}
  Let us compare the two expressions above. We are going to use a
  coupling argument. Coupling two random processes $X, Y$ consists of
  defining a new random process $(X',Y')$ where the first part
  $X'$ is distributed as $X$ and the second
  part $Y'$ is distributed as $Y$, but in addition to that the joint distribution is constructed
  to serve some specific purpose.
  For example, this tool is often used to compare the
  distribution of $X$ and $Y$, see \cite{levin:2009markov}, Chapter 5.

  {\bf Claim 1.} \emph{The two algorithms
    $S^t(\cZ(1),w(1))$ and $S^t(\cZ(1),\vec{1})$ can be coupled so
    that the vector parts of the two are the same.}
  
  We see two push-sum
  algorithms, both initialized with vectors $\cZ(1)$ but with different
  weights. Let us couple them by always choosing the same edge for
  both and also taking the same transmission success decision for the
  two. Notice that during the push-sum steps $S$ the weights $w_i(t)$ have no
  effect on the evolution of the vectors $v_i(t)$. Initially the
  vectors were identical and the same transformations were applied at
  all steps, so they must remain equal for the two algorithms at all times.

  {\bf Claim 2.} \emph{The push sum coefficients $\tau$ and $T(\cZ(1))$ are parallel: More precisely, there
  exists a scalar $\rho>0$ which can be random and dependent from our other
  variables such that
  \begin{equation}
    \label{eq:claim2}
\tau =     T(\cE) \,{\buildrel d \over =}\, T(\cZ(1))\rho,
\end{equation}
where ${\buildrel d \over =}$ means equal in distribution.}

By definition $T(\cZ(1)) = \lim_{t\ra\infty}
  N(S^t(\cZ(1),\vec{1}))$ and we have $T(\cE) = \lim_{t\ra\infty}
  N(S^t(\cZ(1),w(1)))$
  as seen in \eqref{eq:exprTcZSZ1}.
  We can couple $\cZ(1)$ to be the same for the two expressions.
  Then we can use the coupling of Claim 1 from the first step of
  $T(\cZ(1))$ and from the second step of $T(\cE)$.

  We get that the vectors of the realizations of $S^t(\cZ(1),\vec{1})$ and
  $S^t(\cZ(1),w(1))$ are the same for all $t$. Remember now that the operator $N$
  divides the first of the vector by a scalar. Hence the realizations
  of $N(S^t(\cZ(1),\vec{1}))$ and $N(S^t(\cZ(1),w(1)))$ are scalar
  multiple of each other, i.e. they are parallel. Since this is true for
  all $t$, it remains true for their limiting values when $t\to
  \infty$ (note that vectors are bounded away from 0). This implies the
  existence of the $\rho$ mentioned in the claim. It is positive
  because $T(\cE)$ and $T(\cZ(1))$ are in the positive orthant.

  {\bf Claim 3.} \emph{There holds
  \begin{equation}
    \label{eq:invproof1}
    T(\cE) \,{\buildrel d \over =}\, L(T(\cZ(1))).
  \end{equation}}

It follows from Theorem \ref{thm:convergence} that every realization of $\tau = T(\cE)$ is a convex combination of the initial values $\cE = (e_1,e_2,\dots, e_n)$, and has thus a unit 1-norm.
  From \eqref{eq:claim2} we have
  $$\rho \|T(\cZ(1))\| = \|T(\cZ(1))\rho\|   \,{\buildrel d \over =}\, \|T(\cE)\| = 1.$$
  This means that $\rho  \,{\buildrel d \over =}\, \|T(\cZ(1))\|^{-1}$, and thus using the definition of the operator $L$
  $$T(\cE) \,{\buildrel d \over =}\, T(\cZ(1))\rho 
\,{\buildrel d \over =}\, \frac{T(\cZ(1))}{\|T(\cZ(1))\|}
  = L(T(\cZ(1))).$$
  
  {\bf Claim 4.} \emph{The invariance equation (\ref{eq:PSinvariance})
  holds.}

  The distribution of the left hand side expression of
  (\ref{eq:invproof1}) is $\mu$. We
  complete the proof by showing that the right hand side of (\ref{eq:invproof1}) follows the
  distribution indicated on the right hand side of
  (\ref{eq:PSinvariance}). 
  According to Proposition~\ref{prp:PSfn_coeff} the right hand side of
  (\ref{eq:invproof1}) can be expressed using a push sum-coefficient,
  $$L(T(\cZ(1)))= L\langle \tau,\cZ(1)\rangle.$$
  Let us express the distribution of this random variable for any
  realization of $\cZ(1)$. The distribution of $\tau$ is $\mu$, and
  the operators acting on the random variables translate to the
  corresponding measure transformations. We thus arrive at
  $$L^*\langle \mu,\cZ(1)\rangle.$$
  In order to get the overall distribution we have to integrate over
  possible realizations of 
  $\cZ(1)$ which results in
  $$\mE L^*\langle \mu,\cZ(1)\rangle.$$
  This is exactly the expression presented on the right hand side of (\ref{eq:PSinvariance}).

\end{proof}

\section{Proof of error bounds for two agents}
\label{sec:boundproof}

\subsection{Invariance relation for to two agents}

We now focus on the case of 2 agents when both edges are chosen with
probability $1/2$ and the transmission failure probabilities are
equal, simply denoted by $p$.
We assume scalar inputs and we are
interested how far the final consensus value is from the real average.
Our error measure $R$ defined in \eqref{eq:Rdef} has in this case a particularly simple interpretation. One can indeed verify that 
$$R = \mE T\big( (-1,1) \big)^2,$$
that is, it corresponds to the expected square error (and result) when
launching the push-sum with initial values $\cV=(-1,1)$, for which the
average is 0.
We develop lower
and upper bounds for this error. In order to do this, we first investigate
the push-sum coefficient for 2 agents.

Let us interpret previous results from Section~\ref{subsec:measures} for this particular case. 
Observe that the expected value in the invariance relation $\mu =
\mE L^*\left(\langle \mu,\cZ(1)\rangle\right)$ Theorem
\ref{thm:invariance} is taken with respect to $\cZ(1)$, defined as
the vector part or the result of one push-sum step on $\cE$. In the
two-node cases, $\cZ(1)$ can take only four values, corresponding to
successful or failed transmissions from 1 to 2 or from 2 to 1. We
will thus only consider 4 transformations $ L^*\left(\langle
\mu,\cZ(1)\rangle\right)$. Besides, the measure $\mu$ is now
supported on the segment $(1,0)^\top\!-(0,1)^\top\!$ and is thus
inherently a measure on a one dimensional set. 

Hence it will be convenient to interpret it this way: We parametrize the diagonal segment
$(1,0)^\top\!-(0,1)^\top\!$ by the first coordinate and we define $\nu$ the
distribution of this first coordinate. This $\nu$ is then supported on $[0,1]$.
Working with $\nu$ does preserve the symmetry of the problem.
The next Proposition particularizes Theorem \ref{thm:invariance} to
the two-node system in terms of $\nu$:

\begin{proposition}
  \label{prp:invar2nu}
  For the 2-nodes system with equal probability of transmission let
  $\nu$ be the distribution of the first coordinate of the push-sum coefficient.
  There holds 
\begin{equation}
    \label{eq:invar2nu}
    \nu = \frac{1-p}{2}d_1^*(\nu) + \frac{1-p}{2}d_2^*(\nu) + \frac{p}{2}f_1^*(\nu) + \frac{p}{2}f_2^*(\nu),
  \end{equation}
  where $d_1^*, d_2^*, f_1^*, f_2^*$ are the induced measure
  transformations of the functions defined by 
\begin{equation}
  \label{eq:def-d1d2f1f2}
  \begin{aligned}
    d_1(x) &= \frac{1}{3-2x}, & f_1(x) &= \frac{x}{2-x},\\
    d_2(x) &= \frac{2x}{1+2x}, & f_2(x) &= \frac{2x}{1+x}.
  \end{aligned}
\end{equation}
\end{proposition}
Note that the functions $d_1$, $d_2$, $f_1$ and $f_2$ correspond
respectively to successful transmissions initiated by node 1 or 2 and to failed transmissions initiated by 1 or 2.

We show the above Proposition in two steps. First, we express
Theorem \ref{thm:invariance} in the original form, in terms of $\mu$, for the current
case of 2 nodes. Second, we translate the equation for $\mu$ to a
simplified form for $\nu$. 

Consider for instance the case of a successful transmission by node 1
to 2. The initial vectors were $\cE = \big((1,0)^\top\!,
(0,1)^\top\!\big)$ which become $\cZ(1) = \big((1/2,0)^\top\!,
(1/2,1)^\top\!\big)$. For this case, in the spirit of Theorem
\ref{thm:invariance} we use the affine transformation which brings the
segment $(1,0)^\top\! - (0,1)^\top\!$ to $(1/2,0)^\top\! -
(1/2,1)^\top\!$ and compose it with the
central projection mapping it back to $(1,0)^\top\! -
(0,1)^\top\!$ and denote the resulting transformation by
$D_1$. Similarly, $D_2$ is the transformation corresponding to a
delivered message from node 2 to 1.
$F_1,F_2$ are the transformations for failed transmissions by node 1
and 2, respectively. All these are illustrated in Figure~\ref{fig:alltransform}.

The invariance equation of Theorem~\ref{thm:invariance}
becomes
\begin{corollary}
\label{cor:invar2}
For 2 agents the following equation holds for the push-sum coefficient
distribution $\mu$:
\begin{equation}
  \label{eq:invar2}
  \mu = \frac{1-p}{2}D_1^*\mu + \frac{1-p}{2}D_2^*\mu + \frac{p}{2}F_1^*\mu + \frac{p}{2}F_2^*\mu,
\end{equation}
where $D_1^*, D_2^*, F_1^*, F_2^*$ are the induced measure
transformations of the functions $D_1, D_2, F_1, F_2$.
\end{corollary}
\begin{proof}
In the two agent setting there are four cases for the first step
depending on which node is transmitting and whether it is successful or
not. Let us expand the formulation of Theorem~\ref{thm:invariance}.
\begin{align*}
\mu = \mE L^*\left(\langle \mu,\cZ(1)\rangle\right) &= P(1\ra 2\text{
  success})L^*\Big(\big\langle \mu, \mE(\cZ(1)~|~1\ra 2\text{ success})\big\rangle\Big)\\
&+ P(2\ra 1\text{
  success})L^*\Big(\big\langle \mu, \mE(\cZ(1)~|~2\ra 1\text{ success})\big\rangle\Big)\\
&+ P(1\ra 2\text{
  failure})L^*\Big(\big\langle \mu, \mE(\cZ(1)~|~1\ra 2\text{ failure})\big\rangle\Big)\\
&+ P(2\ra 1\text{
  failure})L^*\Big(\big\langle \mu, \mE(\cZ(1)~|~2\ra 1\text{ failure})\big\rangle\Big).
\end{align*}
For the first term, the probability of occurring is
$\frac{1-p}{2}$. Conditioning on the event that the first step was a
$1\ra 2$ successful transmission, we get $\cZ(1) = \big((1/2,0)^\top\!,
(1/2,1)^\top\!\big)$. The measure $\mu$ is supported on the segment
$(1,0)^\top\!-(0,1)^\top\!$, so it is easy to check that the inner product $\big\langle \mu,
\mE(\cZ(1)~|~1\ra 2\text{ success})\big\rangle$ will be connecting $(1/2,0)^\top\!-
(1/2,1)^\top\!$. As the inner product is linear, this measure will be
the same as $\mu$ moved by a similarity transformation to span the
segment between the two new endpoints. The $L^*$ transformation maps back the measure to
$(1,0)^\top\!-(0,1)^\top\!$. 

The other cases are interpreted in a similar way, only the conditional
values of $\cZ(1)$ change. For $2\ra 1$ successful transmission it
is $\big((1,1/2)^\top\!-(0,1/2)^\top\!\big)$, for $1\ra 2$ failed
transmission it is $\big((1/2,0)^\top\!-(0,1)^\top\!\big)$, for $2\ra 1$ failed
transmission it is $\big((1,0)^\top\!-(0,1/2)^\top\!\big)$. 
The series of transformations is illustrated in Figure~\ref{fig:alltransform}.

\begin{figure}[h]
  \centering
  \includegraphics[width=0.8\textwidth]{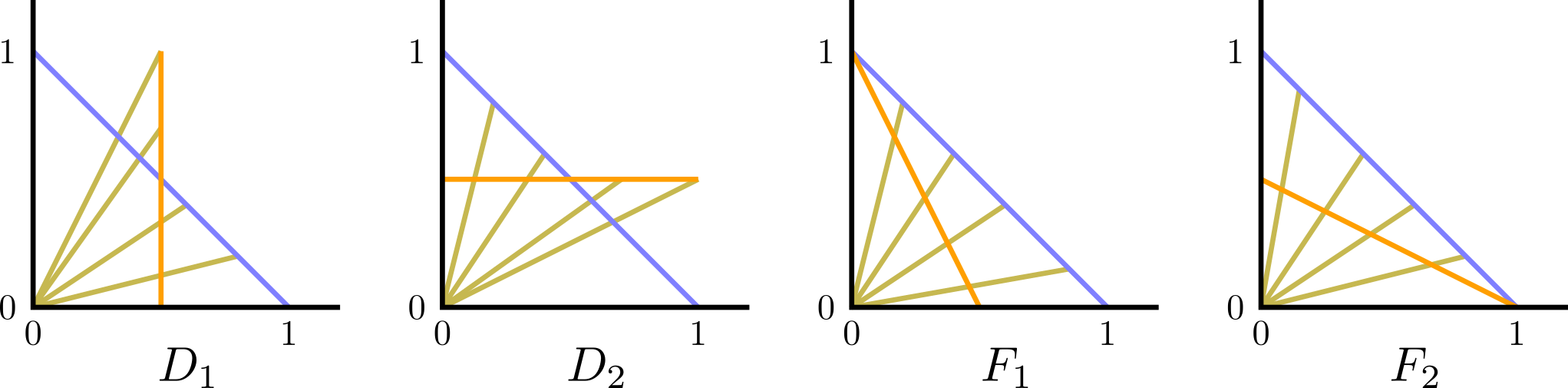}
  \caption{The four transformations of Corollary~\ref{cor:invar2},
    corresponding to possible successful ($D_1,D_2$) and failed ($F_1,F_2$) message transmissions for
    2 nodes. Points on the blue segment
are moved to the orange segment using an affine transformation, then mapped along the yellow lines
back to the blue line using a central projection from the origin.}
  \label{fig:alltransform}
\end{figure}

\end{proof}

It remains to reformulate the equation~\eqref{eq:invar2} in terms of $\nu$.

\begin{proof}[Proof of Proposition~\ref{prp:invar2nu}]
We need to reformulate $D_i,F_i$ in terms of $\nu$, which we
demonstrate for $D_1$. The point $x\in [0,1]$ corresponds to
$(x,1-x)^\top\!$ on the diagonal line supporting $\mu$. For the inner
product we get
$$\left\langle (x,1-x), \big((1/2,0)^\top\!,
(1/2,1)^\top\!\big)\right\rangle = (1/2,1-x)^\top\!.$$
Then using $L$ this is renormalized to have unit 1-norm:
$$\frac{1}{1/2+1-x}(1/2,1-x)^\top\! = \left(\frac{1}{3-2x},\frac{2-2x}{3-2x}\right)^\top\!.$$
Finally, to switch back to the domain of $[0,1]$ we have to take the first
coordinate of the result. In a similar way, we get the four functions
corresponding to the preceding four measure transformations:
\begin{equation}
  \label{eq:def-d1d2f1f2}
  \begin{aligned}
    d_1(x) &= \frac{1}{3-2x}, & f_1(x) &= \frac{x}{2-x},\\
    d_2(x) &= \frac{2x}{1+2x}, & f_2(x) &= \frac{2x}{1+x}.
  \end{aligned}
\end{equation}
Substituting these functions in~\eqref{eq:invar2} while replacing
$\mu$ with $\nu$ we arrive at the statement of the
proposition.
\end{proof}

\subsection{Lower bound}

In this subsection we derive the bounds of Theorem~\ref{thm:lowerbound}.
The general plan is to develop simple necessary conditions on $\nu$
based on (\ref{eq:invar2nu}). We then get a lower bound by optimizing
the error on the broader class of measures only constrained by the
necessary condition obtained.

The first step of simplification is dropping terms from the right hand
side of (\ref{eq:invar2nu}) and using the non-negativity of the terms:
\begin{equation}
  \label{eq:nu_ineq}
  \nu \ge  \frac{p}{2}f_1^*(\nu) + \frac{p}{2}f_2^*(\nu).  
\end{equation}

Let us introduce a useful subdivision of $[0,1]$:
For any $k\in\mN$ let 
\begin{align*}
  h_{-k} &= \frac{1}{2^k+1}, & h_k &= 1- \frac{1}{2^k+1}.
\end{align*}
  
\begin{lemma}
  \label{lm:agrid}
  For any $k\in\mZ$ we have
  \begin{align*}
    f_1(h_k) &= h_{k-1},\\
    f_2(h_k) &= h_{k+1}.
  \end{align*}
\end{lemma}
\begin{proof}
  This can be directly checked using the formulas defining $f_i,h_k$.
\end{proof}

We disregard the rather complicated fine structure of $\nu$ for a
moment and only explore how it acts on the intervals defined by
$h_i$. For this purpose let us define
\begin{equation}
\label{eq:si_ti_def}
\begin{aligned}
s_i &= \nu\big( (h_{-i-1},h_{-i}) \big) = \nu\big( (h_i,h_{i+1})
\big),\\
t_i &= \nu\big( \{h_{-i}\} \big) = \nu\big(\{h_i\}\big).  
\end{aligned}
\end{equation}
Here we use the straightforward symmetry of $\nu$ and the strict monotonicity of $f_i$. Combining
(\ref{eq:nu_ineq}) with Lemma~\ref{lm:agrid} we get the following inequalities.
\begin{lemma}
  \label{lm:s_ineq}
  For $s_i,t_i$ defined in \eqref{eq:si_ti_def} the following inequalities hold:
  \begin{alignat*}{2}
    s_0 &\ge \frac{p}{2}s_0 + \frac{p}{2}s_1, &\qquad &\\
    s_i &\ge \frac{p}{2}s_{i-1} + \frac{p}{2}s_{i+1}, & i &\ge 1,\\
    t_i &\ge \frac{p}{2}t_{i-1} + \frac{p}{2}t_{i+1}, & i &\ge 0.
  \end{alignat*}
\end{lemma}
We derive an even simpler condition on $s_i$ and $t_i$.
  \begin{lemma}
    \label{lm:s_gammarate}
    Let $\phi=\frac{1-\sqrt{1-p^2}}{p}$. For any $i\ge 0$ we have
    \begin{align*}
      \phi s_{i} &\le s_{i+1},\\
      \phi t_{i} &\le t_{i+1}.
    \end{align*}
  \end{lemma}
  \begin{proof}
    We show the inequality for $s_i$, the proof is the same for $t_i$.
    Assume the claim does not hold for a given $i$, meaning $\phi s_i
    > s_{i+1}$. For this proof $i$ is now fixed and $j$ is used as a
    running index.
    We define the auxiliary series $\tilde{s}_{i},\tilde{s}_{i+1},\ldots$ as follows:
    $$\tilde{s}_{i} = s_{i},\qquad \tilde{s}_{i+1} = s_{i+1},$$
    \begin{equation}
      \label{eq:stilde_rec}
      \tilde{s}_j = \frac{2}{p} \tilde{s}_{j-1} - \tilde{s}_{j-2},\qquad j\ge i+2.      
    \end{equation}
    We will show that $\tilde{s}_j\ge s_j$ and $\tilde{s}_j
    \ra -\infty$. This would imply $s_j \ra -\infty$ which is
    impossible for entries of a probability distribution.

    First, to compare the two series, let us define
    $\delta_j = s_j-\tilde{s}_j.$
    Using Lemma~\ref{lm:s_ineq} and the definition of $\tilde{s}_j$ the recursion scheme for this series becomes
    $$\delta_i = 0,\qquad \delta_{i+1} = 0,$$
    $$\delta_j \le \frac{2}{p} \delta_{j-1} - \delta_{j-2},\qquad j\ge i+2.$$
    It is easy to check that $\delta_j\le \delta_{j-1} \le 0$ implies
    $\delta_{j+1}\le \delta_{j} \le 0$. We immediately get $\delta_j\le 0$ in general
    and
    \begin{equation}
      \label{eq:s_stilde}
     \tilde{s}_j \ge  s_j ,\qquad j\ge i.  
    \end{equation}

    Now let us compute the series $(\tilde{s}_j)$. It is defined by a
    second order recursion, so we get the two
    canonical solutions by solving
    $$y^2 = \frac{2}{p} y - 1.$$
    The two roots are
    $$y_1 = \frac{1-\sqrt{1-p^2}}{p} = \phi, \qquad y_2 =
    \frac{1+\sqrt{1-p^2}}{p} = \tilde{\phi}.$$
    These are different except the pathological case of $p=1$.
    Any solution of the recursion (\ref{eq:stilde_rec}) must be a
    mixture of $\phi^j$ and $\tilde{\phi}^j$.
    Then any element of the sequence can be expressed as
    \begin{equation}
      \label{eq:stilde_explicit}
      \tilde{s}_{j} = \lambda\phi^{j-i} + \tilde{\lambda}\tilde{\phi}^{j-i},
    \end{equation}
    for a proper choice of $\lambda, \tilde{\lambda}$. Therefore we need
    to find $\lambda,\tilde{\lambda}$ such that
    \begin{equation}
      \label{eq:alphaequations}
      \begin{aligned}
        \tilde{s}_i &= \lambda + \tilde{\lambda}\\
        \tilde{s}_{i+1} &= \lambda\phi + \tilde{\lambda}\tilde{\phi},
      \end{aligned}
    \end{equation}
    By the condition $\phi s_i > s_{i+1}$ we get
    $\tilde{\lambda}<0$ while solving \eqref{eq:alphaequations}. Adding the fact that $\phi < 1 <
    \tilde{\phi}$, this implies $\tilde{s}_j\rightarrow
    -\infty$.
    
    To sum up, if we had $\phi s_i > s_{i+1}$, then we would get $\tilde{s}_j\rightarrow
    -\infty$, and by \eqref{eq:s_stilde} also $s_j\rightarrow
    -\infty$, which is impossible for probabilities.
  \end{proof}

  To give some intuition, $\phi \approx p/2$ for $p$ near 0, but reaches 1 as
    $p$ increases to 1.
Based on the properties provided by Lemma~\ref{lm:s_gammarate} we are ready to prove Theorem~\ref{thm:lowerbound} claiming \eqref{eq:lowerbound}.

\begin{proof}[Proof of Theorem~\ref{thm:lowerbound}]
  Let us first separate the measure $\nu$ as follows:
  \begin{align*}
    {\nu}_h &= \nu\big|_{\{h_i : i\in\mZ\}},\\
    {\nu}_c &= \nu - \nu_h.
  \end{align*}
  In other words, ${\nu}_h$ corresponds to the $t_i$ terms while
  ${\nu}_c$ corresponds to the $s_i$ terms. For the overall weights we define $M_h =
  \nu_h([0,1])$ and $M_c = \nu_c([0,1])$. Concerning the errors, we
  introduce the following notation for the error corresponding to a
  single point:
  $$r(x) = (1-2x)^2.$$
  We define $R_h,R_c$ for the overall errors of $\nu_h,\nu_c$. As
  we created a separation of the original measure $\nu$ we
  immediately have
  \begin{align*}
    1 &= M_h + M_c,\\
    R &= R_h + R_c.
  \end{align*}

  Let us investigate $\nu_c$. Observe that the error of
  $\nu\big|_{(h_i,h_{i+1})}$ is larger than the error at $h_i$ with
  weight $s_i$ ($i\ge 0$). The symmetrical arguments hold for
  $i<0$. We get
  $$R_c\ge 2\sum_{i=0}^\infty s_i r({h_i}).$$
  We now check how low can the right hand side be while keeping
  the overall weight constant, that is, $2\sum s_i = M_c = \nu_c([0,1])$.

  Let us define the series $(s^*_i)$ such that $2\sum s^*_i =
  \nu_c([0,1])$ and $\phi s^*_i = s^*_{i+1}$ for
  all $i\ge 0$. This can be seen as the extremal series that still
  satisfies the claim of Lemma~\ref{lm:s_gammarate}. We will show that the error corresponding to this series
  is at most that of $(s_i)$. To this end, let us define $\delta_i =
  s_i-s^*_i$. Clearly $\sum \delta_i = 0$. We also have
  $$
  \delta_{i+1} = s_{i+1} - s^*_{i+1} \ge \phi s_i - \phi s^*_i =
  \phi \delta_i.
  $$
  From this we see that $\delta_i>0$ implies that all the later terms of
  $(\delta_i)$ are also positive. Therefore we see the following structure
  of the series $\delta_i$. It has to begin with some non-positive terms up to
  some index $I$ to ensure that the sum is 0. After that, all terms
  are positive.
  
  Furthermore, we see
  \begin{equation}
    \label{eq:di_comparison}
    \sum_{i=0}^I (-\delta_i)r({h_i}) \le \sum_{i=I+1}^\infty \delta_i r({h_i}).
  \end{equation}
  This holds because the sum of the positive weights $-\delta_i$ on the
  left hand side and the sum of $\delta_i$ on the right hand side are the same, but
  the coefficients $r({h_i})$ are larger on the right hand
  side.

  Now let us compare the error of the two series.
  $$\sum_{i=0}^\infty s_i r({h_i}) - \sum_{i=0}^\infty s^*_i r({h_i}) =
  \sum_{i=0}^\infty \delta_i r({h_i}) \ge 0.$$
  The first equality is simply the definition of $\delta_i$, the inequality
  follows from \eqref{eq:di_comparison}. As $(s_i^*)$ is well defined, it
  gives a lower bound on the quadratic error of $\nu_c$.
  \begin{equation*}
    R_c \ge \sum_{i=0}^\infty M_c(1-\phi) \phi^i r({h_i}).
  \end{equation*}

  We can treat $\nu_h$ the same way using $t_i$. Note that now there
  is only one central atom at $h_0$ while there were two intervals of
  interest around $h_0$ for $\nu_c$. Therefore the weights for the
  series $(t_i^*)$ are slightly different, and we get the lower bound
  on the error
  \begin{equation*}
    R_h \ge \frac{1-\phi}{1+\phi}r({h_0}) + \sum_{i=1}^\infty 2M_h\frac{1-\phi}{1+\phi} \phi^i r({h_i}).
  \end{equation*}
  Let us now add up the two lower bounds while using that the error at
  $h_0$ is 0.
  \begin{equation*}
    R \ge \left(M_c+\frac{2M_h}{1+\phi} \right) \sum_{i=1}^\infty (1-\phi)\phi^i
     r({h_i}).
  \end{equation*}
  Knowing that $M_c+M_h=1$ and $\phi\le 1$ this expression is
  minimal if $M_c=1,M_h=0$. With this setting, we arrive at a
  universal lower bound on $R$ as follows.
  \begin{align*}
    R &\ge (1-\phi) \sum_{i=1}^\infty \phi^i
    \left(1-2h_i\right)^2 = (1-\phi) \sum_{i=1}^\infty \phi^i
    \left(\frac{2^i-1}{2^i+1}\right)^2 = (1-\phi) \sum_{i=1}^\infty \phi^i
    \left(1-\frac{4\cdot 2^{i}}{(2^i+1)^2}\right)\\
    &= (1-\phi) \sum_{i=1}^\infty \phi^i
     -(1-\phi) \sum_{i=1}^\infty \phi^i
    \frac{4\cdot 2^{i}}{(2^i+1)^2}
    = \phi - 4(1-\phi) \sum_{i=1}^\infty
    \frac{(2\phi)^{i}}{(2^i+1)^2}.
  \end{align*}
  This is exactly the first lower bound we wanted to show. For the
  second, simpler claim we decrease the denominators $(2^i+1)^2$ to
  $2^{2i}$ in order to get a geometric series that is easy to sum.
  \begin{align*}
    R &\ge \phi - 4(1-\phi)\frac{2\phi}{9}-4(1-\phi) \sum_{i=2}^\infty
    \frac{(2\phi)^{i}}{(2^i+1)^2} \ge \phi - 4(1-\phi)\frac{2\phi}{9}-4(1-\phi) \sum_{i=2}^\infty
    \frac{(2\phi)^{i}}{2^{2i}}\\
    &=\phi - 4(1-\phi)\frac{2\phi}{9}
    -4(1-\phi)\frac{\phi^2}{4}\frac{1}{1-\frac{\phi}{2}}
    = \phi - \frac{8}{9}\phi(1-\phi)-\frac{2}{2-\phi}\phi^2(1-\phi).
  \end{align*}
This is the second lower bound we were aiming for. The asymptotic rate near $p=0$
follows easily for both lower bounds using $\phi =  p/2 + O(p^2)$ for small $p$.
\end{proof}

Having a look at Figure~\ref{fig:lowerupperbounds} we see that the lower
bounds we obtained qualitatively capture the real behavior of the
algorithm. We get an error linear in the failure rate for $p\approx 0$
and then we get an error approaching 1 as $p\approx 1$. Still, we did
some strong simplification steps so quantitatively we do experience a
gap between the simulated and the proven values.

\subsection{Upper bound}

In this subsection we present the proof of the upper bound of
Theorem~\ref{thm:generalupperbound}.
 
\begin{proof}[Proof of Theorem~\ref{thm:generalupperbound}]
Once again we base our studies on the invariance equation
(\ref{eq:invar2nu}). We express $\nu$ as a mixture of measures that are
supported strictly within $[0,1]$.

We give an intuition on the main ideas used in this proof. For a
measure $\pi$ with restricted support in $[x,y]$ for some
$0\le x<y \le 1$, we immediately have an upper bound on the error: we simply
find the point of the interval furthest from $1/2$, and obtain $\max((1-2x)^2,(1-2y)^2)$.

When we apply one of the transformations corresponding to a
transmission on $\pi$, we get a new interval containing the new
support. By studying the evolution of the interval we get an evolving
upper bound.

For a measure $\pi$ that can be expressed as the mixture of different
measures with restricted supports, again we get an upper bound on the
error by taking the weighted average of the error bounds for the
individual measures according to the idea above.

We will now convert these ideas to precise statements and apply them
for the push-sum coefficient measure $\nu$ to get an upper bound on
the error of the push-sum function.

For measures on $[0,1]$ we define the following intervals that will
serve as possible restrictions on the support.
\begin{equation}
  \label{eq:suppintervals}
  \begin{aligned}
    \alpha_i &= \left[ \frac{2^i}{2^{i+1}+1},
      \frac{2^i+1}{2^{i+1}+1}\right],\\
    \beta'_0 &= \left[\frac{1}{5}, \frac{2}{3}\right], \qquad \beta'_1 = \left[0, \frac{2}{3}\right],\\
    \beta''_0 &= \left[\frac{1}{3}, \frac{4}{5}\right], \qquad \beta''_1 = \left[\frac{1}{3}, 1\right],\\
    \gamma_0 &= \left[\frac{1}{5}, \frac{4}{5}\right], \qquad \gamma_1 = \left[0,1\right].
  \end{aligned}
\end{equation}

We construct a related abstract Markov chain. This chain has states
$A_0,A_1,\ldots$ and $B'_0,B'_1, B''_0,B''_1, C_0, C_1$. These will
correspond to possible different supports $\alpha_0,\alpha_1,\ldots$ and
$\beta'_0,\beta'_1, \beta''_0,\beta''_1, \gamma_0, \gamma_1$. From each state
there are four possible transitions with probabilities $(1-p)/2,
(1-p)/2, p/2, p/2$. The transitions are shown in Figure~\ref{fig:ps_trans} and \ref{fig:ps_loss}.
These will correspond to the
four measure transformations used in \eqref{eq:invar2nu}.
For example, whenever $supp(\pi) \subseteq \alpha_0$, we will have $supp(
d_1^*(\pi)) \subseteq \alpha_1$. And indeed, our Markov chain has a
transition from $A_0$ to $A_1$. Moreover, the transition probability
of the Markov chain matches the coefficient of the measure
transformation in \eqref{eq:invar2nu}.

This Markov chain is irreducible, aperiodic and
positive recurrent. Therefore the distribution will approach the unique stationary
distribution from any starting distribution.

For a measure $\sigma$ on $[0,1]$ we define the following operation to get
$\sigma^+$ in the spirit of \eqref{eq:invar2nu}.
\begin{equation}
  \label{eq:transform2nu}
    \sigma^+ = \frac{1-p}{2}d_1^*(\sigma) + \frac{1-p}{2}d_2^*(\sigma) + \frac{p}{2}f_1^*(\sigma) + \frac{p}{2}f_2^*(\sigma),
\end{equation}
where the transformations $d_i^*,f_i^*$ are the ones in \eqref{eq:invar2nu}.

We link the four possible transformations on some measure $\sigma$ with
the transitions of the Markov chain. We say that the state $s$ and the measure $\sigma$
are \emph{consistent} with each other whenever
    \begin{equation}
      \label{eq:statetosupport}
      \begin{aligned}
        &\text{if }  s = A_i \text{ then } supp(\sigma) \subset \alpha_i,\\
        &\text{if }  s = B'_i \text{ then } supp(\sigma) \subset \beta'_i,\\
        &\text{if }  s = B''_i \text{ then } supp(\sigma) \subset \beta''_i,\\ 
        &\text{if }  s = C_i \text{ then } supp(\sigma) \subset \gamma_i.
      \end{aligned}
    \end{equation}
\begin{lemma}
  \label{lm:MCmeasurelink1}
There is a pairing of the transitions of the Markov chain and
    the transformations of measures that correspond to each other:
  Take one of the states $s$ of the Markov chain and a
    measure $\sigma$ on $[0,1]$ that are consistent with each other.
    Let then
    $s^+$ be the next state of the Markov chain after a certain
    transition and $\sigma^+$ the measure resulting from the
    corresponding transformation. Then $s^+$ is consistent with
    $\sigma^+$.
\end{lemma}
\begin{proof}
  This is an elementary but tedious exercise, we only
    present a list of claims to confirm.

    For all pairs $\alpha_i,A_i$ (and similarly for the other states
    and intervals), one has to find the image of $\alpha_i$ under
    the four transformations in \eqref{eq:transform2nu}. Also, one has
    to gather the possible states following $A_i$ for the Markov
    chain. It has to be verified that the images of $\alpha_i$ fall
    within the intervals corresponding to the new state of the Markov
    chain.

    Even more, it should be confirmed that the weights match: the
    measure transformations with coefficients $(1-p)/2$ and $p/2$ in
    \eqref{eq:transform2nu} have to correspond to Markov chain transitions with
    transition probability $(1-p)/2$ and $p/2$.
\end{proof}

\begin{lemma}
  \label{lm:MCmeasurelink2}
  Let $U = u_{A_0},u_{A_1},\ldots$ be a probability
    distribution on the states of the Markov-chain. Let
    $\sigma_{A_0},\sigma_{A_1},\ldots$ be probability measures consistent
    with $A_0,A_1,\ldots$. Define
    $$\sigma = u_{A_0}\sigma_{A_0} + u_{A_1}\sigma_{A_1} + \ldots.$$
    
    Let $U^+ = u^+_{A_0},u^+_{A_1},\ldots$ be the probability distribution when taking one step of
    the Markov chain starting from $U$. Let $\sigma^+$ be the mixture of the four
    transformed version of $\sigma$ according to \eqref{eq:transform2nu}. Then
    there exists probability measures $\sigma^+_{A_0},\sigma^+_{A_1},\ldots$
    consistent with $A_0,A_1,\ldots$ such that
    $$\sigma^+ = u^+_{A_0}\sigma^+_{A_0} + u^+_{A_1}\sigma^+_{A_1} + \ldots.$$
\end{lemma}
\begin{proof}
  We need to apply Lemma~\ref{lm:MCmeasurelink1} multiple times. For example,
    assume for a moment that the Markov chain is at $B'_0$ and we are
    given a $\sigma_{B'_0}$ supported within $\beta'_0$. The Markov chain can
    step to $B'_0, A_0, C_0, B'_1$ with transition probabilities
    $(1-p)/2, (1-p)/2, p/2, p/2$. According to Lemma~\ref{lm:MCmeasurelink1} the four measure
    transformations corresponding to these four transitions will lead
    to four measures supported within $\beta'_0, \alpha_0, \gamma_0,
    \beta'_1$.

    When we relate the setting to the full Markov chain and $\sigma$, we
    observe that the probability of the Markov chain being in $B'_0$
    is $u_{B'_0}$ and so is the weight of $\sigma_{B'_0}$ in the
    expression of $\sigma$.

    After a step in the Markov chain, the probability distribution at
    any node, say $A_0$, is the aggregated probability of the incoming
    transitions, from $B'_0,B'_1,B''_0, B''_1$. In the same way
    $\sigma^+_{A_0}$ is the mixture of measures coming from the
    corresponding transformations of
    $\sigma_{B'_0},\sigma_{B'_1},\sigma_{B''_0}, \sigma_{B''_1}$. Observe that the
    initial probability distribution of the Markov chain and the
    weights in the expression of $\sigma$ match. Even more, the
    transition probabilities and the mixture weight are the
    same. Therefore the new probability distribution of the Markov
    chain an the mixture weights of the new measure must agree as well.
\end{proof}

\begin{figure}[ht]
\begin{center}
  \includegraphics[width=0.8\textwidth]{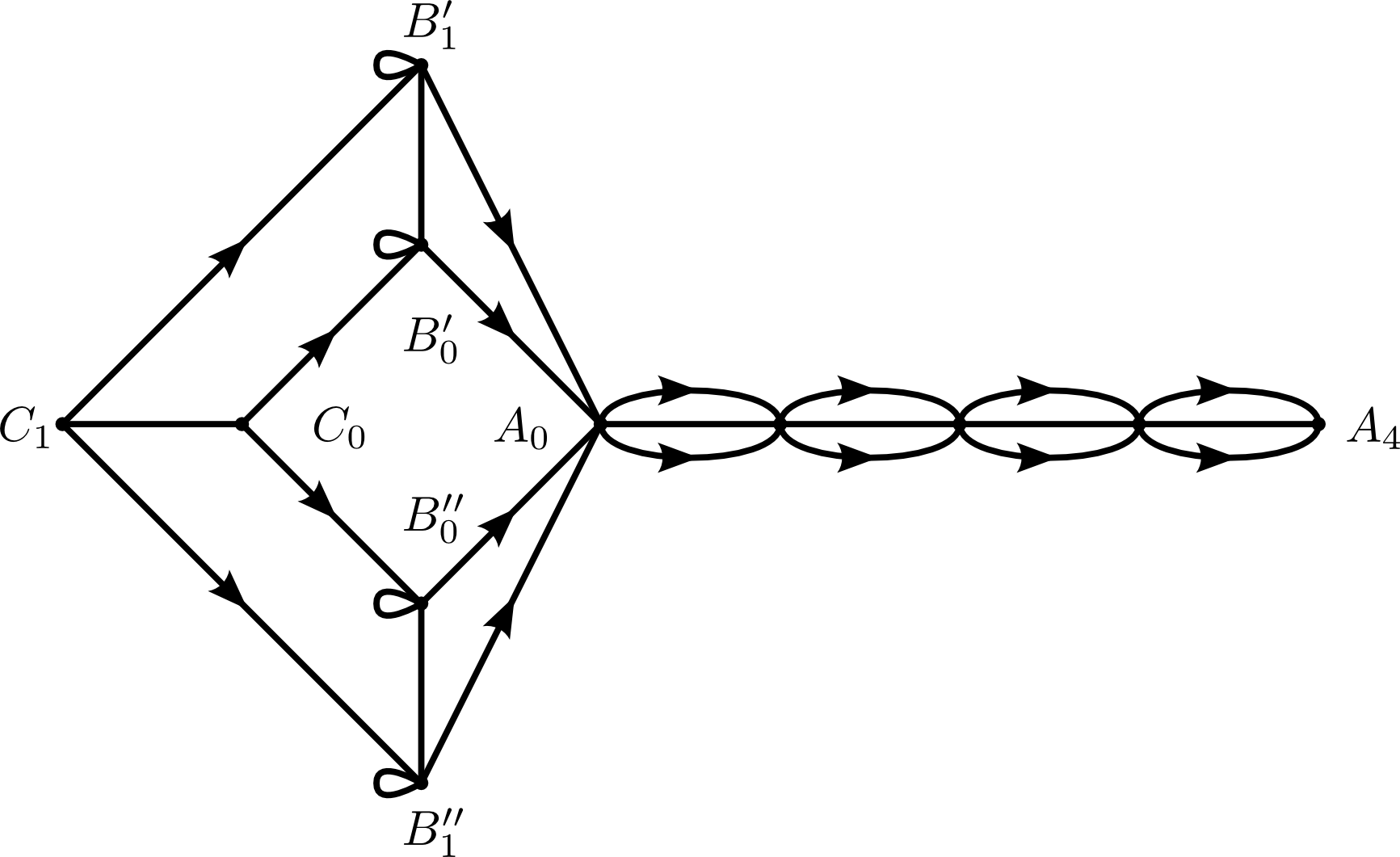}
\end{center}
\caption{Transitions
with probability $(1-p)/2$ for the Markov chain used in the proof of Theorem~\ref{thm:generalupperbound}.}
\label{fig:ps_trans}
\end{figure}

\begin{figure}[ht]
\begin{center}
  \includegraphics[width=0.8\textwidth]{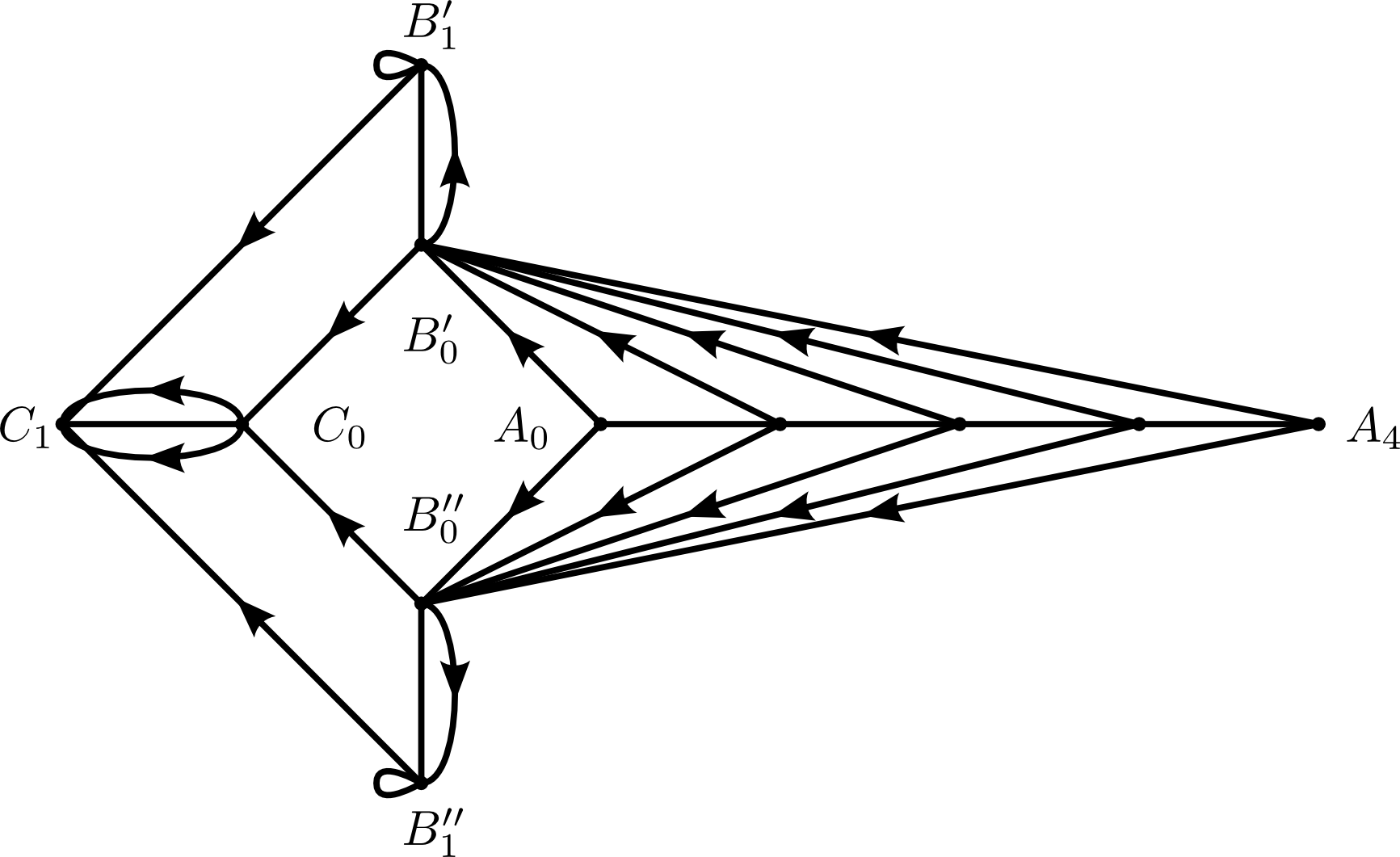}
\end{center}
\caption{Transitions with
  probability $p/2$ for the Markov chain used in the proof of Theorem~\ref{thm:generalupperbound}.}
\label{fig:ps_loss}
\end{figure}

Let us now initiate the process with $\sigma = \nu$, the push-sum coefficient measure,
and with the probability distribution $\mathbbm{1}_{C_1}$ for the Markov
chain putting all weight on $C_1$. This $\sigma$ is indeed supported within $[0,1]$, so this is a
consistent choice. Remixing $\sigma$ according to \eqref{eq:transform2nu} does
not change it as it is the solution of the invariance equation \eqref{eq:invar2nu}
mentioned before. On the other hand, the Markov chain will approach its
unique stationary distribution step by step. According to Lemma~\ref{lm:MCmeasurelink2} we get that $\nu$ can be expressed as the mixture
of measures supported within $\alpha_0,\alpha_1,\ldots$ with the weights being the values
of the stationary distribution of the Markov chain.

We now calculate this stationary distribution. Let $a_i, c_i$ be
the stationary probabilities of being in states $A_i, C_i$. Using the
symmetry of $B'_i$ and $B''_i$ let $b_i$ the stationary probabilities
of being in states $B'_i$ or $B''_i$.
Also,
let us define the total weights of the different
types of states as
\begin{align*}
  S_A &= \sum_{i=0}^\infty a_i,\\
  S_B &= b_0+b_1,\\
  S_C &= c_0+c_1.
\end{align*}
Looking at the scheme of possible transitions, we immediately find
that
\begin{align*}
  S_A &= (1-p)S_A + \frac{1-p}{2}S_B,\\
  S_B &= p S_A + \frac{1}{2}S_B + (1-p)S_C,\\
  S_C &= \frac{p}{2}S_B + p S_C.
\end{align*}
Knowing also that these three sum up to one, we get
$$S_A = (1-p)^2,\quad S_B=2p(1-p), \quad S_C = p^2.$$
For $a_i$ we see the simple recursion $a_i=(1-p)a_{i-1}$. Taking into
account their sum $S_A$ we get
$$a_i = p(1-p)^{i+2}.$$
For the other nodes, we have the equations.
\begin{align*}
  b_0 &= p S_A + \frac{1-p}{2} b_0 + (1-p) c_0,\\
  c_0 &= \frac{p}{2} b_0.
\end{align*}
These finally lead to
\begin{align*}
  b_0 &= \frac{2p(1-p)^2}{1+p^2},\\
  b_1 &= \frac{2p^2(1-p^2)}{1+p^2},\\
  c_0 &= \frac{p^2(1-p)^2}{1+p^2},\\
  c_1 &= \frac{2p^3}{1+p^2}.
\end{align*}

We get an upper bound on the error $R$ of the measure $\nu$ if we
combine upper bounds for the different intervals $\alpha_i,\beta'_i,\beta''_i,\gamma_i$ with the
weights of the stationary distribution of the Markov-chain.
Recall that we have an error bound for a certain interval if we find
the furthest point $y$ from $1/2$ and then evaluate $(1-2y)^2$ for the
quadratic
error. We will use $r$ to denote these
errors for the different intervals. By omitting the obvious calculations we get
\begin{align*}
  r({\alpha_i}) &= \frac{1}{(2^{i+1}+1)^2},\\
r({\beta'_0}) = r({\beta''_0}) &= \frac{9}{25},\\
r({\beta'_1}) = r({\beta''_1}) &= 1,\\
  r({\gamma_0}) &= \frac{9}{25},\\
  r({\gamma_1}) &= 1.
\end{align*}

Finally, let us combine all our estimates.
\begin{align*}
  R &\le \sum_{i=0}^\infty a_i r({A_i}) + b_0 r({B_0}) + b_1 r({B_1}) +
  c_0 r({C_0}) + c_1 r({C_1}) \\
  &= \sum_{i=0}^\infty \frac{p(1-p)^{i+2}}{(2^{i+1}+1)^2} +
  \frac{9}{25}\frac{2p(1-p)^2}{1+p^2} + 1\frac{2p^2(1-p^2)}{1+p^2} +
  \frac{9}{25}\frac{p^2(1-p)^2}{1+p^2} + 1\frac{2p^3}{1+p^2}.\\
  &\le \sum_{i=0}^\infty \frac{p(1-p)^{i+2}}{(2^{i+1})^2} +
  \frac{p}{25(1+p^2)}\left(18+23p+50p^2-41p^3\right)\\
  &= \frac{p(1-p)^2}{4}\frac{1}{1-(1-p)/4} +
  \frac{p}{25(1+p^2)}\left(18+23p+50p^2-41p^3\right)\\
  &= \frac{p(1-p)^2}{3+p} +
  \frac{p}{25(1+p^2)}\left(18+23p+50p^2-41p^3\right).
\end{align*}
This is the bound presented in the claim of the theorem.
\end{proof}

\section{Conclusions}
\label{sec:conclusions}

We have analyzed the push-sum algorithm in the presence of
transmission failures. This algorithm was originally designed to
perform perfect averaging on a network with directed
communication. When transmission failures are possible, the values of
the nodes of
the network still converge to a common value, but this might not be
the exact average of the initial measurements.

The final value is a random variable determined by the sequence of
communication steps and the sequence of transmission failures. We
develop new tools to better understand the resulting value,
and we form an equation that implicitly describes the distribution of this
random variable. Further investigation is performed for the simple case when there are
only two nodes, we develop lower and upper bounds on the expected
error.

There are very natural follow-up questions to consider for future
research. 

For the case of two nodes, the error bounds do still have a
gap between them, there is still room for improvement. One way of
  achieving this could be to consider a different relaxation of our invariance
  relation \eqref{eq:invar2nu}, a stronger version of \eqref{eq:nu_ineq} on
  which our bounds of Theorem
  \ref{thm:generalupperbound} are built. Inequality \eqref{eq:nu_ineq} does
  indeed not take into account the effect of successful
  transmissions as it only contains the operations
  $f_1^*,f_2^*$. Consequently it is also
  valid for other update rules in case of successful
  transmission, and so is Theorem \ref{thm:generalupperbound}.

  The other challenge would be to adapt our methodology to networks with multiple
nodes. We already have an insight on the distribution of the final
value, but it is not straightforward how this could be translated into
quantitative bounds.

\section*{Acknowledgments}
We would like to thank Asuman Ozdaglar for the inspiring discussion
that eventually lead to this research question
being investigated.

\bibliographystyle{ieeetr}
\bibliography{ringmixing,pushsum}

\end{document}